\documentclass[aps,pra,preprint]{revtex4-1}
\usepackage[latin1]{inputenc}
\usepackage[english]{babel}
\usepackage{amsfonts,amssymb,amsmath,amscd,mathtools,slashed}
\usepackage{amsthm}
\usepackage{enumerate}
\usepackage{graphicx}
\usepackage{color}
\usepackage{verbatim}
\usepackage{psfrag}
\usepackage{epsfig}
\usepackage{tikz}
\usepackage{mathrsfs}

\newcounter{mnotecount}[section]
\renewcommand{\themnotecount}{\thesection.\arabic{mnotecount}}
\newcommand{\mnote}[1]%{}
{\protect{\stepcounter{mnotecount}}$^{\mbox{\footnotesize
$%\!\!\!\!\!\!\,
\bullet$\themnotecount}}$ \marginpar{%\color{red}%
\raggedright\tiny\em
$\!\!\!\!\!\!\,\bullet$\themnotecount: #1} }

\definecolor{darkgreen}{rgb}{0,.5,0}

\newtheorem{thm}{Theorem}
\newtheorem{Thm}{Theorem}

\newtheorem{Cor}[thm]{Corollary}

\newtheorem{Prop}[thm]{Proposition}

\theoremstyle{definition}

\newtheorem{Def}[thm]{Definition}

\newtheorem{Rem}[thm]{Remark}

\newtheorem{Ex}[thm]{Example}

\newtheorem{principle}{Principle}

\newcommand{\norm}[1]{\left\Vert {#1} \right\Vert}
\newcommand{\abs}[1]{\left\vert {#1} \right\vert}

\newcommand{\sR}{{\mathbb R}}
\newcommand{\sC}{{\mathbb C}}

\newcommand{\B}{\mathcal{B}}

\renewcommand{\H}{\mathcal{H}}

\newcommand{\M}{\mathcal{M}}

\newcommand{\K}{\mathcal{K}}

\newcommand{\CC}{\mathbb{C}}
\newcommand{\NN}{\mathbb{N}}
\newcommand{\R}{\mathcal{R}}

\newcommand{\X}{\mathcal{X}}

\newcommand{\E}{\mathcal{E}}

\newcommand{\bmu}{\boldsymbol\mu}

\DeclareMathOperator{\diag}{diag}

\DeclareMathOperator{\supp}{supp}
\DeclareMathOperator{\sech}{sech}
\DeclareMathOperator{\argmax}{arg \, max}
\DeclareMathOperator{\ba}{\mathbf{a}}

\newcommand{\dt}{\partial}
\newcommand{\x}{\times}
\newcommand{\vc}{\vcentcolon =}
\newcommand{\cv}{= \vcentcolon}
\DeclareMathOperator{\id}{id}
\newcommand{\<}{\left\langle}		%% bra
\renewcommand{\>}{\right\rangle}		%% ket

\renewcommand{\R}{\mathbb{R}}
\newcommand{\PP}{\mathscr{P}}
\newcommand{\Bs}{\mathscr{B}}

\newcommand{\wt}{\widetilde}

\begin{document}

%\preprint{}

\title{Causal evolution of wave packets}


\author{Micha{\l} Eckstein}
\email{michal.eckstein@uj.edu.pl}
%\thanks{}
\affiliation{Faculty of Physics, Astronomy and Applied Computer Science, Jagiellonian University, ul. prof. Stanis{\l}awa {\L}ojasiewicza 11, 30-348 Krak\'ow, Poland}
\affiliation{Copernicus Center for Interdisciplinary Studies, ul. S{\l}awkowska 17, 31-016 Krak\'ow, Poland}

\author{Tomasz Miller}
\email{T.Miller@mini.pw.edu.pl}
\affiliation{Faculty of Mathematics and Information Science, Warsaw University of Technology, ul. Koszykowa 75, 00-662 Warsaw, Poland.}
\affiliation{Copernicus Center for Interdisciplinary Studies, ul. S{\l}awkowska 17, 31-016 Krak\'ow, Poland}


\date{\today}

\begin{abstract}
Drawing from the optimal transport theory adapted to the relativistic setting we formulate the principle of a causal flow of probability and apply it in the wave packet formalism. We demonstrate that whereas the Dirac system is causal, the relativistic-Schr\"odinger Hamiltonian impels a superluminal evolution of probabilities. We quantify the causality breakdown in the latter system and argue that, in contrast to the popular viewpoint, it is not related to the localisation properties of the states.
\end{abstract}

%\keywords{}

\maketitle

\section{\label{sec:intro}Introduction}

Causality, understood as the impossibility of superluminal transfer of information, is considered one of the fundamental principles, which should be satisfied in any physical theory. Whereas it is readily
implemented in classical theories based on Lorentzian geometry, the status of causality in quantum theory was controversial from its dawn. As expressed in the famous Einstein--Podolsky--Rosen paper
\cite{EPR}, the main stumbling block is the inherent nonlocality of quantum states. However, quantum nonlocality on its own cannot be utilised for a superluminal transfer of information, neither can quantum correlations be communicated between spacelike separated regions of spacetime \cite{QIandGR}. In fact, the principle of causality can be invoked to discriminate theories that predict stronger than quantum correlations \cite{InformationCausality}. %In quantum field theory, where nonlocality is even more prevailing, microscopic causality is promoted to one of the~axioms \cite{Haag,Wightman}.

It is usually argued that the proper framework to study causality in quantum theory should be that of quantum field theory (see for instance \cite{Yngvason,EntanglementCreation,FermiSuperconductors}). Moreover, some researchers conclude that causality --- seemingly broken in one-particle relativistic quantum mechanics --- is magically restored at the QFT level \cite{MinimalPacket,PeskinSchroeder,Thaller}. On the other hand, the results of \cite{WSWSG11} suggest that if a relativistic quantum system is acausal before the second quantisation, then this drawback cannot be cured by the introduction of antiparticles.

From the viewpoint of quantum field theory, the wave packet formalism gives a phenomenological rather than fundamental description of Nature. Nevertheless, it serves as a handful approximation commonly used in atomic, condensed matter \cite{RQM_apps,TopInsulators} and particle physics \cite{BernardiniFlavour,Beuthe}. Regardless of the adopted simplifications, its statistical predictions confronted in the experiments cannot be at odds with the principle of causality.

Within the wave packet formalism, one can investigate the status of causality in course of the evolution of the system, driven by a relativistically invariant Hamiltonian \cite{WSWSG11}. This firstly requires a precise definition, which accurately disentangles the nonlocality of quantum states from the causality violation effects as, for instance, interference fringes can travel with superluminal speed, but cannot be utilised to transfer information \cite{Berry2012}. The results usually invoked in this context are these of Hegerfeldt \cite{Hegerfeldt1} (see also \cite{Hegerfeldt1985,HegerfeldtFermi,Hegerfeldt2001,Hegerfeldt2}), which show that an initially localised \footnote{`Localised' in the~context of Hegerfeldt's theorem usually means compactly supported in space, but the~argument extends to states with exponentially bounded tails
\cite{Hegerfeldt1985}.} quantum state with positive energy immediately develops infinite tails. Hegerfeldt's approach, however, faced criticism \cite{Yngvason} based on the impossibility of preparing a `localised' state \cite{ReehSchlieder} (compare \cite{ZitterQFT} though). It is usually concluded that Hegerfeldt's theorems, which are mathematically correct, provide an alternative argument against the localisation of quantum relativistic states \cite{MinimalPacket,Barat2003,Yngvason} rather than a `proof of acausality'.

Whereas from Hegerfeldt's theorem it follows that locality and positive energy of a quantum state necessarily imply superluminal probability flow, the use of a nonlocal initial state does not a priori guarantee a causal evolution.  In fact, to our best knowledge, no rigorous definition of causality in the wave packet formalism has been provided, beyond the case of states with
exponentially bounded tails. Moreover, there seems to be no reason to restrict the studies to positive-energy wave packets only, as for instance in the Dirac-like systems in atomic and condensed matter physics superpositions of positive and negative energy states are routinely involved \cite{ZitterNature,ZitterSemiconductors}. %\ME{Expand citations}.

The aim of this paper is to study the issue of causality in the wave packet formalism for states with arbitrary localisation properties. To this end we employ the notion of causality for Borel probability measures developed in our recent articles \cite{EcksteinMiller2015,Miller16}. %In fact, the framework adopted in \cite{EcksteinMiller2015} allows taking into account various effects lying beyond the standard wave packet approach, involving curved spacetimes and states spread in the timelike direction -- see Section \ref{sec:outlook}. %We shall restrict ourselves to the continuous evolution processes postponing the discussion of abrupt changes of the state quantum state (`\textit{quantum jumps}')to a forthcoming paper.
Armed with a rigorous notion of causality suitable for the study of arbitrary wave packets, we investigate the status of causality during the evolution of two relativistic quantum systems, driven respectively by the Dirac and relativistic-Schr\"odinger Hamiltonians. We demonstrate that in the Dirac system, the evolution of \emph{any} initial wave packet is causal, even in the presence of interactions. On the other hand, the propagation under the relativistic-Schr\"odinger Hamiltonian turns out to be at odds with the principle of causality. % The acausal behaviour in the latter system is not, however, visible for every initial state.
We confirm and clarify the conclusions of Hegerfeldt concerning the acausal behaviour of exponentially localised states with positive energy. In addition, we provide explicit examples of quantum states with heavy tails, that do not fulfil Hegerfeldt's localisation assumption, but do break the principle of causality. We quantify the acausal effects and confirm their transient character, detected in \cite{WSWSG11} for compactly supported initial states. We therefore conclude that in the relativistic-Schr\"odinger system Einstein's causality is indeed violated, but the latter is a feature of the Hamiltonian and not of any particular state.

The paper is organised as follows: In Section \ref{sec:omega} we present the basic definition of causality for probability measures from \cite{EcksteinMiller2015} and the physical intuition behind. Therein, we also coin the definition of a \emph{causal evolution} and discuss its Lorentz invariance. Then, in Section \ref{sec:QM}, we apply the developed theory in the wave packet formalism. After some general considerations concerning the quantification of causality breakdown, we turn to the $n$-dimensional Dirac system and show that it impels a causal evolution of probability measures, regardless of the choice of the initial spinor. This result holds also when, possibly non-Abelian, external gauge field is minimally coupled to the system. Then, we take a closer look at the relativistic-Schr\"odinger system in 2 dimensions. We
confirm the breakdown of causality in the course of evolution of an initial Gaussian state, derived in \cite{Hegerfeldt1985} and checked also in \cite{WSWSG11}. Next, we turn to states with exponentially bounded tails and show, via explicit examples, that Hegerfeldt's bound is superficial. Finally, we demonstrate the violation of causality for wave packets of power-like decay. %On the other hand, we demonstrate that there exists nonlocal initial wave packets which seem to evolve causally in the relativistic-Schr\"odinger system.
A summary of our work, together with further comparison with Hegerfeldt's theorem, comprises Section \ref{sec:outlook}. Therein, we also make an outlook into the potential empirical implications of our results and their possible refinements.

\section{\label{sec:omega}Causality for probability measures}

\subsection{\label{subsec:general}The causal relation}

We start with a brief summary of the main concepts contained in \cite{EcksteinMiller2015}. This requires some notions from Lorentzian geometry, topology and measure theory, which we invoke without introducing the complete mathematical structure behind. For a detailed exposition on these topics the reader is referred to standard textbooks on general relativity \cite{Beem,Penrose1972,Wald} and optimal transport theory \cite{ambrosio2008gradient,Villani2008} or, simply, to the `Preliminaries' section in \cite{EcksteinMiller2015}.

Let $\M$ be a spacetime. For any $p,q \in \M$ we say that $p$ \emph{causally precedes} $q$ (denoted $p \preceq q$) iff there exists a piecewise smooth causal curve $\gamma: [0,1] \rightarrow \M$, such that $\gamma(0) = p$ and $\gamma(1) = q$. It is customary to denote the set of causally related pairs of events by $J^+$, i.e. $J^+ := \{ (p,q) \in \M^2 \ | \ p \preceq q \}$. For any $p \in \M$ one defines the \emph{causal future (past)} of $p$ via
\begin{align*}
J^+(p) := \{ q \in \M \ | \ p \preceq q \} \qquad \left( J^-(p) := \{ r \in \M \ | \ r \preceq p \} \right).
\end{align*}
Similarly, for any set $\X \subseteq \M$ one denotes $J^\pm(\X) := \bigcup\limits_{p \in \X} J^\pm(p)$. %Furthermore, a subset $\X \subseteq \M$ is called \emph{achronal} iff no two points in $\X$ can be connected with a future-directed timelike curve.

%In the following, $\M$ will always denote a spacetime with a metric $g$ of signature $(- + \ldots +)$. Recall that a piecewise smooth curve $\gamma: [0,1] \rightarrow \M$ is \emph{causal} iff its tangent vector $\gamma'$ satisfies $g(\gamma'(\tau), \gamma'(\tau)) \leq 0$ for all $\tau$ for which it is defined.
%With the aid of the notion of a causal curve, one endows $\M$ with a binary relation $\preceq$ called the \emph{causal precedence relation}. Concretely, for any $p,q \in \M$ we say that $p$ \emph{causally precedes} $q$ (denoted $p \preceq q$) iff there exists a causal curve $\gamma: [0,1] \rightarrow \M$ going from $p$ to $q$, i.e. such that $\gamma(0) = p$ and $\gamma(1) = q$.
%Recall now that a probability measure $\mu$ is called Borel iff it is defined on the $\sigma$-algebra $\Bs(\M)$ generated by all open subsets of $\M$. We denote

Let us now consider $\PP(\M)$ -- the set of all Borel probability measures on $\M$ (which we shall simply call `measures' from now on), i.e. measures defined on the $\sigma$-algebra $\Bs(\M)$ of all Borel subsets of $\M$, and normalised to 1. In particular, $\PP(\M)$ contains all measures of the form $\rho \cdot \lambda_\M$, where $\rho$ is a probability density on $\M$ and $\lambda_\M$ is the standard Lebesgue measure on $\M$. Also, one can regard $\M$ as naturally embedded in $\PP(\M)$, the embedding being the map $p \mapsto \delta_p$, where the latter denotes the Dirac measure concentrated at the event $p$.

In \cite{EcksteinMiller2015} we demonstrated that the causal relation $\preceq$ extends in a natural way from the spacetime $\M$ onto $\PP(\M)$ \cite[Definition 2]{EcksteinMiller2015}. Concretely, we have:

\begin{Def}\cite{EcksteinMiller2015}
\label{causality_def_true}
Let $\M$ be a~spacetime. For any $\mu,\nu \in \PP(\M)$ we say that $\mu$ \emph{causally precedes} $\nu$ (symbolically $\mu \preceq \nu$) iff there exists $\omega \in \PP(\M^2)$ such that
\begin{enumerate}[i)]
\item $\omega(A \times \M) = \mu(A)$ and $\omega(\M \times A) = \nu(A)$ for any $A \in \Bs(\M)$,
\item $\omega(J^+) = 1$.
\end{enumerate}
Such an~$\omega$ is called a~\emph{causal coupling} of $\mu$ and $\nu$.
\end{Def}

The above definition mathematically encodes the following physical intuition: The existence of a joint probability measure $\omega$ provides a (possibly non-unique) probability flow from $\mu$ to $\nu$ and the condition $\omega(J^+) = 1$ says that the flow is conducted exclusively along future-directed causal curves. We shall denote the set of all couplings between $\mu, \nu \in \PP(\M)$ (i.e. joint probability measures satisfying $i)$) by $\Pi(\mu,\nu)$ and the set of causal ones by $\Pi_c(\mu,\nu)$.

%Notice that the above definition requires $J^+$ be an element of $\Bs(\M)$, which is indeed the case \cite[Theorem 4]{EcksteinMiller2015}. Nevertheless, in the following we shall restrict ourselves only to the (still very broad) class of the so-called \emph{causally simple spacetimes}, whose definition includes the requirement that $J^+$ be a \emph{closed} subset of $\M^2$ (the other requirement is the absence of \emph{causal loops}).

%If $J^+$ is closed, then the condition $\omega(J^+)=1$ can be equivalently expressed as $\supp \, \omega \subseteq J^+$ \cite[Remark 5]{EcksteinMiller2015}.

In spacetimes equipped with a sufficiently robust causal structure one has the following characterisation of the causal precedence relation:
\begin{Thm}
\label{charcond}
Let $\M$ be a~causally simple spacetime \footnote{Causally simple spacetimes are slightly more general than the globally hyperbolic ones. In particular, they do not contain closed causal curves and admit a global time function, but they do not, in general, admit a Cauchy hypersurface. For a precise definition see \cite{MS08}.} and let $\mu, \nu \in \mathfrak{P}(\M)$. Then, $\mu \preceq \nu$ if and only if for all compact $\K \subseteq \supp \, \mu$
\begin{align}
    \label{charcond2}
    \mu(\K) \leq \nu(J^+(\K)).
\end{align}
\end{Thm}
\begin{proof}
On the strength of \cite[Theorem 8]{EcksteinMiller2015}, $\mu$ causally precedes $\nu$ iff for all compact $C \subseteq \M$
\begin{align}
\label{charcond1}
\mu(J^+(C)) \leq \nu(J^+(C)),
\end{align}
which trivially implies (\ref{charcond2}). In order to show the converse implication, let $C \subseteq \M$ be any compact set. Recall that every measure on $\M$, icluding $\mu$, is \emph{tight}, i.e. the $\mu$-measure of any Borel subset of $\M$ can be approximated from below by $\mu$-measures of its compact subsets. In particular,
\begin{align*}
    \forall \, \varepsilon > 0 \ \exists \, \K_{\varepsilon} \subseteq J^+(C) \cap \, \supp \, \mu \ \textnormal{ compact and such that } \ \mu(J^+(C) \cap \supp \, \mu) \leq \mu(\K_{\varepsilon}) + \varepsilon
\end{align*}
\noindent
%Observe that since $J^+(C)$ is closed \cite[Proposition 3.68]{MS08}, therefore $\K_{\varepsilon}$ is also compact in $\M$ and so inequality (\ref{charcond1}) holds for it.
Using (\ref{charcond2}), one thus can write that
\begin{align*}
\mu(J^+(C)) & = \mu(J^+(C) \cap \supp \, \mu) \leq \mu(\K_{\varepsilon}) + \varepsilon \leq \nu(J^+(\K_{\varepsilon})) + \varepsilon
\\
& \leq \nu(J^+(J^+(C) \cap \supp \, \mu)) + \varepsilon \leq \nu(J^+(J^+(C))) + \varepsilon = \nu(J^+(C)) + \varepsilon,
\end{align*}
\noindent
which yields $ii)$ as soon as one takes $\varepsilon \rightarrow 0^+$.
\end{proof}

%Typically, the initial measure $\mu$ is localised on a time-slice (see Section \ref{sec:outlook} however). In this case, condition \eqref{charcond2} reads $\mu(C) \leq \nu(J^+(C))$ for all compact $C \subseteq \supp \, \mu$.

%%%%%%%%%%%%%%%%%%%%%%

\begin{comment}

\begin{proof}
The equivalence of $i) \Leftrightarrow ii)$ is proven in \cite[Theorem 8]{EcksteinMiller2015}. The implication $ii) \Rightarrow iii)$ is trivial. We still need to prove, conversely, that $ii) \Leftarrow iii)$.

To this end, let $\K \subseteq \M$ be any compact set. Recall that the $\mu$-measure of any Borel set can be approximated from below by $\mu$-measures of its compact subsets. In particular,
\begin{align*}
    \forall \, \varepsilon > 0 \ \exists \, C_{\varepsilon} \subseteq J^+(\K) \cap \supp \, \mu \textnormal{ compact and such that } \mu(J^+(\K) \cap \supp \, \mu) \leq \mu(C_{\varepsilon}) + \varepsilon
\end{align*}
\noindent
Observe that since $J^+(\K)$ is closed \cite[Proposition 3.68]{MS08}, therefore $C_{\varepsilon}$ is also compact in $\supp \, \mu$ and inequality (\ref{charcond2}) holds for it. One thus can write that
\begin{align*}
\mu(J^+(\K)) & = \mu(J^+(\K) \cap \supp \, \mu) \leq \mu(C_{\varepsilon}) + \varepsilon \leq \mu(J^+(C_{\varepsilon})) + \varepsilon \leq \nu(J^+(C_{\varepsilon})) + \varepsilon
\\
& \leq \nu(J^+(J^+(\K) \cap \supp \, \mu)) + \varepsilon \leq \nu(J^+(\K)) + \varepsilon,
\end{align*}
\noindent
which yields $ii)$ as soon as one takes $\varepsilon \rightarrow 0^+$.
\end{proof}

\end{comment}

%%%%%%%%%%%%%%%%%%%%%

Condition (\ref{charcond2}) provides a link with the `no-signalling' intuition behind the principle of causality. Indeed, imagine that there exists a physical process, which implies a probability flow $\mu \rightsquigarrow \nu$ --- i.e. there exists $\omega \in \Pi(\mu,\nu)$ --- which is superluminal, i.e. $\omega(J^+)<1$. Then, Theorem \ref{charcond} says that there exists a compact region of spacetime $\K$, such that the probability leaks out of its future cone. In this case, an observer localised in $\K$ could encode some information in a probability measure $\mu$, for instance by collapsing a non-local quantum states of a larger system, and transfer it to a recipient beyond $J^+(\K)$ -- the causal future of $\K$. Such a method of signalling would be rather inefficient, due to its statistical nature, but would be a priori possible (compare similar arguments given in \cite{Hegerfeldt2001} or \cite{Hegerfeldt1985}).

If $\M$ is causally simple, then the condition $\omega(J^+)=1$ can be equivalently expressed as $\supp \, \omega \subseteq J^+$ \cite[Remark 5]{EcksteinMiller2015}. This, in particular, implies the following necessary condition for the causal precedence of two measures \cite[Proposition 5]{EcksteinMiller2015}.

\begin{Prop}\cite{EcksteinMiller2015}
\label{neccond}
Let $\M$ be a~causally simple spacetime and let $\mu, \nu \in \mathfrak{P}(\M)$, with $\mu$ compactly supported. If $\mu \preceq \nu$, then $\supp \, \nu \subseteq J^+(\supp \, \mu)$.
\end{Prop}

In other words, if the~measure $\mu$ is compactly supported, then the~support of any $\nu$ causally preceded by $\mu$ should lie within the~future of $\supp \, \mu$. Whereas this condition is necessary, it is not sufficient, even in the case of both $\mu$ and $\nu$ compactly supported. This is readily illustrated by the following counterexample:

\begin{figure}[h]
\begin{tikzpicture}
%\draw[gray, thin] (-6,0) -- (6,0);
%\draw[gray, thin] (-6,3) -- (6,3);
\draw[black, thick] (-2,0) -- (2,0);
\draw[black] (-4,3) -- (4,3);
\draw[black,ultra thick] (-1.5,0) -- (-0.5,0);
\draw[black,ultra thick] (3,3) -- (4,3);
\draw[black, dashed] (-2,0) -- (-5,3);
\draw[black, dashed] (2,0) -- (5,3);
\draw[black, dotted] (-0.5,0) -- (2.5,3);
\draw[black, dotted] (-1.5,0) -- (-4.5,3);
\node[] at (0,-0.5) {$\mu$};
\node[] at (-1,-0.5) {$\K$};
\node[] at (0,3.5) {$\nu$};
\node[] at (3.5,3.5) {$\K'$};
\node[] at (2.2,1.5) {$J^+(\supp \mu)$};
\node[] at (-1,1.5) {$J^+(\K)$};

\end{tikzpicture}
\caption{Although $\supp \, \nu$ lies in the future of $\supp \, \mu$, the excessive weight condensed in the region $\K$ cannot flow causally to $\K'$.}
\end{figure}
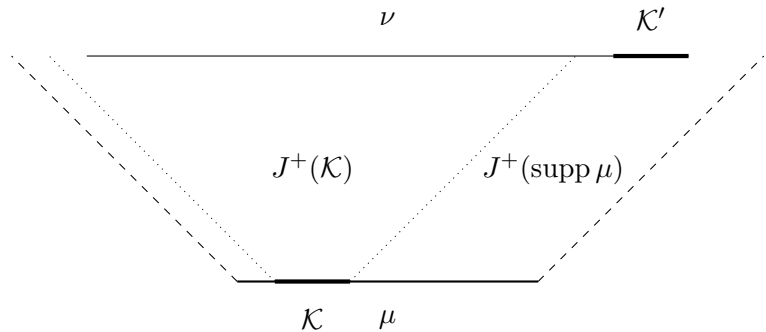

\subsection{\label{sec:evolution}The causal dynamics of measures}

The formalism developed in \cite{EcksteinMiller2015} and summarised above establishes the kinematical structure of $\PP(\M)$. We shall now formalise the requirement that any evolution of probability measures should respect the inherent causal structure. This task has been accomplished in \cite{Miller16} in full generality of curved spacetimes. Since the main objective of this paper is the application in wave packet formalism, we will focus exclusively on the Minkowski spacetime and assume the measures to be localised in time, i.e. concentrated on parallel time-slices.% We shall come back to the general framework in Section \ref{sec:outlook}.

Let us fix an interval $I \subseteq \sR$ and consider a measure-valued map
\begin{align*}%\label{E}
\E: I \to \PP(\R^n), \quad t \mapsto \E(t) \cv \mu_t,
\end{align*}
which describes a time-dependent probability measure on $\sR^n$. This map can be equivalently regarded as a family of measures $\{\bmu_t\}_{t \in I} \subseteq \PP(\M)$, where $\bmu_t := \delta_t \times \mu_t$ and $\M := \sR^{1+n}$ denotes the $(1+n)$-dimensional Minkowski spacetime. One can think of the map $t \mapsto \mu_t$ as a \emph{curve in $\PP(\sR^n)$} parametrised by $t \in I$. If $\mu_t = \delta_{x(t)}$, then one recovers a curve $t \mapsto x(t)$ in $\sR^n$, whereas $\bmu_t = \delta_{(t,x(t))}$ becomes the corresponding worldline in $\M$ of a classical point particle. We shall refer to the map $\E$, or equivalently to the corresponding family $\{\bmu_t\}$, as the \emph{dynamics of measures} or \emph{evolution of measures}.% or the \emph{probability flow}.

The compatibility of the dynamics of measures with the causal structure of $\PP(\M)$ is formalised in the following definition:
\begin{Def}\label{def:causal_evolution}
We say that an evolution of measures is \emph{causal} iff
\begin{align}
\label{causal_evo}
\forall \, s,t \in I  \text{ with } s \leq t \qquad \bmu_s \preceq \bmu_t,
\end{align}
in the sense of Definition \ref{causality_def_true}.
\end{Def}

One may be concerned about the apparent frame-dependence of thus defined (causal) evolution of measures. Indeed, the measures $\bmu_t$ live on $t$-slices, and so this way of describing the dynamics of a non-local phenomenon manifestly depends on the slicing of the spacetime associated with the chosen time parameter. To put it differently, consider two observers $O$ and $O'$, one Lorentz-boosted with respect to the other, who want to describe the dynamics of the same non-local phenomenon. Their evolutions of measures $\E$ and $\E'$, respectively, employ two different time parameters $t$ and $t'$ and, consequently, two different collections of time slices. In particular, it is a priori not clear whether $O$ and $O'$ would always agree on the causality of their respective evolution of measures.

This matter has been thoroughly analysed in \cite[Section 5]{Miller16}, in a much broader class of spacetimes. It turns out that, in spite of the apparent frame-dependence of Definition \ref{def:causal_evolution}, the property of the evolution of measures being causal is independent of the choice of the time parameter. %However, the explanation why this is the case is rather technical, involving many nontrivial facts from topology, mathematical relativity and optimal transport theory. As such, it goes far beyond the scope of this paper.
Interested reader can find all the details in \cite{Miller16}.

$\PP(\sR^n)$-valued maps can be utilised to model various physical entities evolving according to some dynamics. The most natural examples concern classical spread objects, such as charge or energy densities (see Section \ref{sec:examples}). In the present paper, we demonstrate that the same concept can be successfully applied to probability measures obtained from wave functions in the position representation.

As stressed in the introduction, the wave packet formalism has a phenomenological character from the viewpoint of relativistic quantum theory. Moreover, in actual experiments the measured probabilities are affected by the characteristic of the detector \cite{EberlyWodkiewicz,WodkiewiczPRL}. Therefore, it is more adequate to speak of causality of the \emph{model} rather then the quantum system itself. The latter is believed to be causal par excellence, on the strength of the micro-causality axiom of quantum field theory \cite{Haag,Wightman}.

%In the present paper, we will primarily be interested in the probability measures obtained from the wave-functions in the position representation. However, our framework can be equally well applied to model purely classical spread objects (see Section \ref{sec:examples}).

%To determine the status of causality in a given physical system, we need to consider all possible evolutions of measures allowed by the system's dynamics. This motivates the following, somewhat restrictive, definition.

\begin{Def}\label{def:causal_sys}
We say that the \emph{model of a physical system is causal} iff any evolution of measures on $\R^n$ governed by its dynamics is causal in the sense of Definition \ref{def:causal_evolution}.
\end{Def}

%If the evolution in $\PP(\sR^n)$ is deterministic, as in the case of dust or energy density dynamics (see Section \ref{sec:examples}), then Definition \ref{def:causal_sys} takes into account all possible initial configurations $\mu_{t_0}$. If, however, the measures are \emph{secondary} quantities of the system derived, for instance, from the wave function dynamics, the evolution $t \mapsto \mu_t$ is not determined uniquely by the initial measure $\mu_{t_0}$. \ME{Do we need that at all?}%On the other hand, if the system exhibits invariance under some symmetry (Poincar\'e symmetry in particular), which lifts to the elements of $\PP(\M)$, then s

Equipped with the rigorous definition of a causal evolution we can express the demand of causality of the statistical predictions of any physical model.

\begin{principle}\label{post}
Any description of a physical system, which involves an evolution of probability measures on $\sR^n$ must be causal in the sense of Definition \ref{def:causal_sys}.
\end{principle}

%After some further general considerations and examples from classical physics we will study the status of the formulated principle in the wave packet formalism. We will demonstrate that it correctly disentangles the inherent non-locality and superposition effects of the waves functions from the flow of probability.

\subsection{\label{subsec:con_eq}Continuity equation}

%The aim of this subsection is to show that the $n$-dimensional Dirac equation (on the Minkowski space-time) implies a \emph{causal evolution} of the initial state.

%However, crucial in the entire reasoning will not be the Dirac equation itself, but rather the \emph{continuity equation}, which can be obtained from it. Strictly speaking, we actually show that
%\emph{any} theory in which states evolve in accordance with a continuity equation with a subluminal velocity field is causal in the sense presented in this work.

%Let us now restrict ourselves to the case of $\M$ being an $(n+1)$--dimensional Minkowski spacetime. We shall consider physical quantities localised in time, i.e. modelled by probability measures of the form $\delta_t \times \mu_t$, where $\mu_t \in \PP(\mathbb{R}^n)$ at any time instant $t \in \R$.

In physics one often encounters the \emph{continuity equation}, which describes the transport (or the \emph{flow}) of a certain conserved quantity, described by a density
function $\rho: [0,T] \times \mathbb{R}^n \rightarrow \mathbb{R}$. Typically, the equation has the form
\begin{align}
\label{conteq}
\frac{\partial}{\partial t} \rho + \nabla_x \cdot \textbf{j} = 0,
\end{align}
\noindent
for (sufficiently regular) $\rho$ and a time-dependent vector field $\textbf{j}: [0,T] \times \mathbb{R}^n \rightarrow \mathbb{R}^n$ called the \emph{flux} of $\rho$. If there is a \emph{velocity field} $\textbf{v}$, according to which the flow runs (as it happens for instance in
fluid mechanics), then $\textbf{j} = \rho \textbf{v}$.

The aim of this section is to show that any theory, in which the distribution of a physical quantity evolves in accordance with a continuity equation with a subluminal velocity field, is causal in the sense of Definition \ref{def:causal_sys}.

We begin with the definition of the continuity equation in the space of measures, as given e.g. in \cite[Definition 1.4.1]{Crippa2007PhD}.

\begin{Def}[\cite{Crippa2007PhD}]
Let $I = [0,T]$, for some $T>0$. We say that an evolution of measures $\E: t \mapsto \mu_t$ satisfies the continuity equation with a given time-dependent Borel velocity field
$\textnormal{\textbf{v}}: [0,T] \times \mathbb{R}^n \rightarrow \mathbb{R}^n$, $(t,x) \mapsto \textnormal{\textbf{v}}_t(x)$ iff
\begin{align}
\label{conteq0}
\frac{\partial}{\partial t} \mu_t + \nabla_x \cdot \left( \textnormal{\textbf{v}}_t \mu_t \right) = 0
\end{align}
\noindent
holds in the distributional sense, i.e. for all $\Phi \in C^\infty_c((0,T) \times \mathbb{R}^n)$,
\begin{align}
\label{conteq1}
& \int\limits_0^T \int\limits_{\mathbb{R}^n} \left[\frac{\partial \Phi}{\partial t} + \textnormal{\textbf{v}}_t \cdot \nabla_x \Phi \right] d\mu_t dt = 0 \, .
\end{align}
\end{Def}
The continuity equation allows one to regard the time-dependent measure $\mu_t$ as some sort of a fluid. Its density flows, but overall constitutes a conserved quantity.  %Intuitively speaking, such fluid can neither appear nor disappear, also locally (i.e. it cannot ``teleport'').
Its `particles' (fluid parcels) move according to the velocity field $\textbf{v}$ in a continuous manner. One would intuitively expect that if the flow of measures is to behave reasonably, the magnitude of $\textbf{v}$ should be bounded. This expectation is attested by the following following theorem \cite[Theorem 3]{Bernard12} (see also \cite[Theorem 3.2]{ambrosio2008continuity} or \cite[Theorem 6.2.2]{Crippa2007PhD} for other formulations). %\ME{Is this a good reference? Does $\vert\cdot\vert$ mean $\norm{\cdot}$ on $\R^n$?}

\begin{Thm}[\cite{Bernard12}]\label{superposition}
Let $T>0$ and denote $\Gamma_T := C([0,T],\mathbb{R}^n)$. Let $\E$ satisfy the continuity equation with velocity field $\textnormal{\textbf{v}}$ such that
\begin{align}
\label{assumpt}
\exists \, V>0 \ \, \forall \, (t,x) \in [0,T] \times \mathbb{R}^{n} \quad \norm{\textnormal{\textbf{v}}_t(x)} \leq V.
\end{align}
Then, there exists a measure $\sigma \in \PP\left( \Gamma_T \right)$ such that:
\begin{itemize}
\item $\sigma$ is concentrated on absolutely continuous curves $\gamma \in \Gamma_T$ satisfying
\begin{align}
\label{ODE}
\dot{\gamma}(t) = \textnormal{\textbf{v}}_t(\gamma(t)) \qquad \textrm{for } t \in (0,T) \textrm{  a.e.};
\end{align}
\item $\left(\textnormal{ev}_t\right)_\ast \sigma = \mu_t$ for every $t \in [0,T]$, where $\textnormal{ev}_t: \Gamma_T \rightarrow \mathbb{R}^n$ denotes the evaluation map $\textnormal{ev}_t(\gamma) = \gamma(t)$.
%\begin{align}
%\label{superpos}
%\forall \, \phi \in C_b(\mathbb{R}^n) \quad \int\limits_{\mathbb{R}^n} \phi \, d\mu_t = \int\limits_{\mathbb{R}^n} \left( \int\limits_{\Gamma_T} \left( \phi \circ \textnormal{ev}_t \right)(\gamma)
%d\sigma_x(\gamma) \right) d\mu_0(x).
%\end{align}
\end{itemize}
\end{Thm}
One can say that the measure $\sigma$ prescribes a family of curves along which the infinitesimal `parcels' flow during the evolution. Since we put very little requirements on $\textbf{v}$ (namely, that it is Borel and bounded), curves satisfying (\ref{ODE}) might cross each other and the measure $\sigma$ itself is in general not unique.

One would intuitively expect that the probability flow is causal if the norm of the velocity field governing its dynamics is bounded by the the speed of light $c$ at every point of $\M$. The following theorem shows that this is indeed the case.

%We will now build a bridge between the continuity equation and the causal evolution in $(1+n)$-dimensional Minkowski spacetimes.
\begin{Thm}
\label{bridge} %Let $\mathbb{R}^{1+n}$ be equipped with the Minkowski metric $\eta$, $\eta_{\mu\nu} dx^\mu dx^\nu = -c^2 dt^2 + dx_1^2 + \ldots + dx_n^2$.
Let $T>0$ and let the evolution of measures $\E$ satisfy the continuity equation with a velocity field $\textnormal{\textbf{v}}$ such that
\begin{align}
\label{assumpt1}
\forall \, (t,x) \in [0,T] \times \mathbb{R}^{n} \quad \norm{\textnormal{\textbf{v}}_t(x)} \leq c.
\end{align}
Then, $\E$ is causal in the sense of Definition \ref{def:causal_evolution}.
%, for all $s,t \in [0,T], s \leq t$
%\begin{align}
%\label{bridge_thesis}
%\delta_s \times \mu_s \ \preceq \ \delta_t \times \mu_t.
%\end{align}
\end{Thm}
\begin{proof}
By (\ref{assumpt1}), there exists a measure $\sigma \in \PP\left( \Gamma_T \right)$ with the properties listed in Theorem \ref{superposition}.

We claim the following: For every absolutely continuous curve $\gamma \in \Gamma_T$ satisfying (\ref{ODE}), we have
\begin{align}
\label{bridge1}
(s,\gamma(s)) \preceq (t, \gamma(t)), \quad 0 \leq s \leq t \leq T.
\end{align}

Note that the curve $t \mapsto (t, \gamma(t))$, being absolutely continuous, has tangent vectors $(1, \gamma^\prime(t))$ for almost all $t \in (0,T)$. Moreover, these tangent vectors are causal by
(\ref{assumpt1}). However, this curve need not be piecewise smooth, so (\ref{bridge1}) does not follow (that) trivially.

On the other hand, in the Minkowski spacetime (\ref{bridge1}) is equivalent to the inequality
\begin{align}
\label{bridge2}
\norm{\gamma(t) - \gamma(s)} \leq c(t-s), \quad 0 \leq s \leq t \leq T
\end{align}
\noindent
and this can be easily proven by means of the fundamental theorem of calculus, which is valid precisely for absolutely continuous functions. Namely, we can write
\begin{align*}
\forall \, s,t \in [0,T] \quad \gamma(t) = \gamma(s) + \int\limits_s^t \gamma^\prime(\tau) d\tau.
\end{align*}
\noindent
Therefore, if $s \leq t$, then
\begin{align*}
\norm{\gamma(t) - \gamma(s)} = \norm{ \int\limits_s^t \gamma^\prime(\tau) d\tau } \leq \int\limits_s^t \norm{ \gamma^\prime(\tau) } d\tau = \int\limits_s^t \norm{ \textbf{v}_\tau(\gamma(\tau)) } d\tau
\leq c(t-s),
\end{align*}
\noindent
where in the last inequality we employed (\ref{assumpt1}), thus proving (\ref{bridge2}) and, consequently, (\ref{bridge1}).

Now, for any $s,t \in [0,T]$, $s \leq t$ define the map $\textnormal{Ev}_{(s,t)}: \Gamma_T \rightarrow \M^2$ by $\textnormal{Ev}_{(s,t)}(\gamma) := \left( (s,\gamma(s)), (t,\gamma(t)) \right)$. We claim that $\omega := \left(\textnormal{Ev}_{(s,t)}\right)_\ast \sigma$ is a causal coupling of $\bmu_s$ and $\bmu_t$.

Indeed, for any $A \in \Bs(\M)$, using its characteristic function $\chi_A$, one can write
\begin{align*}
\omega(A \times \M) = \int\limits_{\M^2} \chi_A(p) d\omega(p,q) = \int\limits_{\Gamma_T} \chi_A(s,\gamma(s)) d\sigma(\gamma) = \int\limits_{\sR^n} \chi_A(s,y) d\mu_s(y) = \bmu_s(A).
\end{align*}
\noindent
One similarly shows that $\omega(\M \times A) = \bmu_t(A)$.

To demonstrate $\omega(J^+) = 1$, notice that we have
\begin{align*}
\omega(J^+) = \int\limits_{\M^2} \chi_{J^+} d\omega = \int\limits_{\Gamma_T} \underbrace{\chi_{J^+}\left( (s,\gamma(s)), (t,\gamma(t)) \right)}_{= \, 1} d\sigma(\gamma) = \int\limits_{\Gamma_T} d\sigma = 1,
\end{align*}
where we made use of \eqref{bridge1}. This concludes the proof of $\omega$ being a causal coupling and, by the arbitrariness of $s,t$, we have thus shown that the evolution $\E: t \mapsto \mu_t$ is causal.

% DO WE NEED WHAT IS BELOW?

%Let now $f \in {\cal C}(\mathbb{R}^{1+n})$. Using the superposition principle and (\ref{bridge1}), we obtain that
%\begin{align*}
%\int\limits_{\mathbb{R}^{1+n}} f \, d\left(\delta_t \times \mu_t\right) & = \int\limits_{\mathbb{R}^n} f(t,y) d\mu_t(y) = \int\limits_{\mathbb{R}^n} \left( \int\limits_{\Gamma_T} f(t,\gamma(t))
%d\sigma_x(\gamma) \right) d\mu_0(x)
%\\
%& \geq \int\limits_{\mathbb{R}^n} \left( \int\limits_{\Gamma_T} f(s,\gamma(s)) d\sigma_x(\gamma) \right) d\mu_0(x) = \int\limits_{\mathbb{R}^n} f(s,y) d\mu_s(y) = \int\limits_{\mathbb{R}^{1+n}} f \,
%d\left(\delta_s \times \mu_s\right),
%\end{align*}
%\noindent
%where the inequality follows from the causality of $f$. Since $f$ was an arbitrary element of ${\cal %C}(\mathbb{R}^{1+n})$ and $s,t \in [0,T], s \leq t$ were also arbitrary, we have just proven that the %family $\{\mu_t\}_{t \in [0,T]}$ is causal.
\end{proof}

As a corollary of Theorem \ref{bridge}, we unravel the following relation between the continuity equation for probability densities (\ref{conteq}) and the causality of their flow.
\begin{Cor}
\label{contcor}
Let $T>0$ and let $\rho,\, \textnormal{\textbf{j}}$ satisfy equation
\eqref{conteq}. Suppose, additionally, that $\rho \geq 0$ % (or $\rho \leq 0$)
and that $\int\limits_{\mathbb{R}^n} \rho(0,x) dx =: Q \in (0,+\infty)$. % (and hence, by (\ref{conteq}),$\int\limits_{\mathbb{R}^n} \rho(t,x) dx = Q$ for any $t \in [0,T]$).
Then, if $J := (c \rho, \textnormal{\textbf{j}})$ is a causal vector field on the Minkowski spacetime $\M := \mathbb{R}^{1+n}$, then
the evolution $\E : t \mapsto \mu_t$ with $d\mu_t(x) \vc \frac{\rho(t,x)}{Q}d^nx$ is causal.
\end{Cor}
\begin{proof}
Note that \eqref{conteq} guarantees that $\int\limits_{\mathbb{R}^n} \rho(t,x) dx = Q$ for any $t \in [0,T]$ and the definition of $\mu_t$ is sound.

Now, observe that $\E$ satisfies the continuity equation \eqref{conteq0} with the velocity field $\textbf{v} = (v^k)_{k=1,\ldots,n}$ defined as
\begin{align*}
\forall \, (t,x) \in [0,T] \times \mathbb{R}^{n} \quad v_t^k(x) := \left\{ \begin{array}{ll} \frac{j^k(t,x)}{\rho(t,x)}, & \textrm{for } (t,x) \textrm{ such that } \rho(t,x) \neq 0 \\ 0, & \textrm{for }
(t,x) \textrm{ such that } \rho(t,x) = 0 \end{array} \right. .
\end{align*}
Indeed, for any $\Phi \in C^\infty_c((0,T) \times \mathbb{R}^n)$ one has (we employ Einstein's summation convention),
\begin{multline*}
\int\limits_0^T \int\limits_{\mathbb{R}^n} \left[\frac{\partial \Phi}{\partial t} + \textbf{v}_t \cdot \nabla_x \Phi \right] d\mu_t dt = \frac{1}{Q} \int\limits_0^T \int\limits_{\mathbb{R}^n}
\frac{\partial \Phi}{\partial t} \rho \, d^nx dt + \frac{1}{Q} \int\limits_0^T \int\limits_{\mathbb{R}^n} \rho \, v_t^k \frac{\partial \Phi}{\partial x^k} \, d^nx dt
\\
= - \frac{1}{Q} \int\limits_0^T \int\limits_{\mathbb{R}^n} \Phi \frac{\partial \rho}{\partial t} \, d^nx dt - \frac{1}{Q} \int\limits_0^T \int\limits_{\mathbb{R}^n} \Phi \frac{\partial j^k}{\partial
x^k} d^nx dt = - \frac{1}{Q} \int\limits_0^T \int\limits_{\mathbb{R}^n} \Phi \underbrace{\left[ \frac{\partial \rho}{\partial t} + \frac{\partial j^k}{\partial x^k} \right]}_{= \, 0 \textnormal{ by }
(\ref{conteq})} d^nx dt = 0
\end{multline*}
\noindent
and so condition \eqref{conteq1} is satisfied.

In remains now to check that condition \eqref{assumpt1} holds, which amounts to proving that for all $(t,x) \in [0,T] \times \mathbb{R}^{n}$,
\begin{align*}
\norm{ \textbf{j}(t,x) } \leq c \abs{ \rho(t,x) }.
\end{align*}
But the latter is precisely the condition for the vector field $J := (c \rho, \textnormal{\textbf{j}})$ to be causal, which is true by assumption.
\end{proof}

\subsection{\label{sec:examples}Examples from classical physics}

Corollary \ref{contcor} shows that Definition \ref{def:causal_evolution} correctly encodes the common intuitions concerning the causal flow, at least in the domain of classical physics. Before we move to the quantum realm, let us provide further evidence in favour of Principle \ref{post} by invoking concrete examples.

% shows that the definition of causality for nonlocal phenomena studied in this work agrees with the common intuitions. Many physical theories that are generally regarded as causal indeed \emph{are} causal in the sense studied here. Before we demonstrate the powerfulness of the intr

%\ME{Recast the examples!}

\begin{Ex}
\label{Ex2}
By Maxwell's equations, if $\rho$ and $\textbf{j}$ denote, respectively, the \emph{charge density} and the \emph{current density} (on $\mathbb{R}^3$), then they satisfy the continuity equation \eqref{conteq}. It is
well known that $J := (c \rho, \textbf{j})$ is a causal four-vector field \cite[\S 28]{Landau}.

Suppose that $\rho \geq 0$ or $\rho \leq 0$ and that the total charge $Q$ is finite. Then, Corollary \ref{contcor} assures that the evolution of $\rho$ is causal.
\end{Ex}

\begin{Ex}
\label{Ex3}
Consider a time- and space-dependent electromagnetic field $\textbf{E}$, $\textbf{B}$. In the absence of external charges and currents, the electromagnetic energy density $u := \frac{1}{2}\left( \varepsilon_0
\norm{\textbf{E}}^2 + \frac{1}{\mu_0} \norm{\textbf{B}}^2 \right)$ satisfies the continuity equation
\begin{align*}
\frac{\partial}{\partial t} u + \nabla_x \cdot \textbf{S} = 0,
\end{align*}
\noindent
where $\textbf{S} := \frac{1}{\mu_0} \textbf{E} \times \textbf{B}$ is the Poynting vector.

As is well known, the quadruple $(c u, \textbf{S})$ is a causal four-vector field, which is actually equal to $c T^{\mu 0}$, where $T^{\mu \nu}$ constitutes the stress--energy tensor of the electromagnetic field \cite[\S\S 32--33]{Landau}. If we now assume that the total energy $\int_{\mathbb{R}^3} u(0,x) dx$ is finite, Corollary \ref{contcor} guarantees that $u$ evolves causally.
\end{Ex}

\begin{Ex}\label{Ex4}
Generalising the previous example, consider a stress--energy tensor $T^{\mu\nu}$ satisfying the \emph{dominant energy condition} (DEC) \cite{malament}:
\begin{align*}
& X^\mu \textnormal{ is a causal vector field }
\\ & \quad \Rightarrow \quad T^{\mu\nu} X_\mu X_\nu \geq 0 \quad \wedge \quad T^{\mu\nu} X_\nu \textnormal{ is a causal vector field.}
\end{align*}
Then, $T^{00} \geq 0$ and the vector field $T^{\mu 0}$ is causal, as is clear by taking $X := (1,0,\ldots,0)$.

The energy conservation principle takes the form (in the Minkowski spacetime) of the continuity equation $\partial_\mu T^{\mu 0} = 0$. All that, together with Corollary \ref{contcor}, implies that
the energy density $\rho := T^{00}$ evolves causally, provided that the total energy $\int_{\mathbb{R}^n} \rho(0,x) dx$ is finite.
\end{Ex}

\section{\label{sec:QM}The wave packet formalism}

We have illustrated the techniques from the optimal transport theory on classical examples. Now we will argue that the same concept proves useful in the quantum theory described via the wave packet formalism. The first hint in favour of this claim is provided by Example \ref{Ex3}: It was observed by Bia{\l}ynicki-Birula \cite{Birula94,Birula96,Birula96_2} that the energy density of the electromagnetic field admits a probabilistic interpretation and can be written as the modulus square of the photon wave function. Example \ref{Ex3}, on the strength of Corollary \ref{contcor}, immediately implies that the description of the one-particle quantum electromagnetism via photon wave function impels a causal probability flow and thus harmonises with Principle \ref{post}. Let us stress that this result, although clearly based on the Lorentz invariance of Maxwell equations, is not trivial. The wave function, being a complex object, induce interference effects in the probability density, which could in principle spoil the causal flow of probability. The fact that this is not the case shows that Definition \ref{def:causal_evolution} correctly disentangles causality violation from the quantum superposition effects.

Since the concept of a photon wave function is in close analogy with the Dirac formalism, it is natural to expect that the latter also enjoys Principle \ref{post}. This is indeed the case, as we will shortly show (see Section \ref{subsec:Dirac}). Before doing so, let us establish the general framework for the study of causality in wave packet formalism on the $(1+n)$-dimensional Minkowski spacetime.

We assume that the quantum system at hand is described by the wave function $\psi : \R^{1+n} \to \CC^k$ for some $k \in \NN$, evolving under the Schr\"odinger equation
\begin{align*}
i \dt_t \psi(t,x) = \hat{H} \psi(t,x),
\end{align*}
where $\hat{H}$ is the Hamiltonian operator. We shall adopt the natural units $\hbar = c =1$.

As the wave function $\psi$ is normalised to 1 at any instant of time, it defines a probability density $\norm{\psi(t,x)}^2$ on $\R^n$ for every $t \in \R$. By fixing a time interval $[0,T]$ we obtain an evolution of measures $\E: t \mapsto \mu_t$, with $d\mu_t(x) = \norm{\psi(t,x)}^2 d^n x \in \PP(\R^n)$. Equipped with Definition \ref{def:causal_evolution} we can thus rigorously study the issue of causality during the evolution of a given quantum system.

Let us note that the evolution of measures $\mu_t$ is \emph{not} uniquely determined by the initial measure $\mu_0$, as initial wave functions differing by a (non-constant) phase factor will yield the same initial probability distribution $\mu_0$, but different evolutions.

%We shall mainly be interested in the quantum evolution driven by relativistically-invariant Hamiltonians. Working in this context immediately raises the question of the Lorentz invariance. It is natural to assume that the probability measures transform as scalars, i.e. $\Lambda_* \mu_t(x) = \mu_{\Lambda(t)}(\Lambda(x))$, which is indeed the case in the systems we consider below. Note, however, that there is no straightforward relationship between the family $\{\mu_t\}_{t \in [0,T]}$  and $\{\Lambda_*(\mu_t)\}_{t \in [0,T]}$, as the transformed one lives on different time-slices. Fortunately, Theorem \ref{thm:Lorentz} guarantees that if the evolution of measures is causal in one frame, it will be so in any other frame.

\subsection{\label{subsec:quantify}Quantifying the breakdown of causality}

As pointed out in \cite{WSWSG11}, it is desirable to have a quantitative picture of causality breakdown in a given system. In fact, Hegerfeldt's result is only qualitative (see Section \ref{subsec:Hegerfeldt}). %, as it does not tell anything about how large is the ball, to which the probability leaks too fast and `how much' of this probability actually leaks in \TM{All this is nowhere `above'}.
It might thus happen, that in a given quantum system, the acausal probability flow is in fact irrelevant, as, for instance, the space-scale of causality violation lies well below or well above the scale of validity of the wave packet formalism. Moreover, the results of \cite{WSWSG11} show that the causality breakdown in the relativistic-Schr\"odinger system is a transient effect and it becomes marginal rather quickly.

To quantify the scale of causality breakdown, the notion of the `outside probability' was introduced in \cite{WSWSG11}. In our notations, it can be written as
\begin{align}
N(t,\mu_0) & = \bmu_t\left( \M \setminus J^+(\supp \bmu_0) \right) = 1 - \bmu_t\left(J^+(\supp \bmu_0) \right) \notag \\
& = (\bmu_0 - \bmu_t)\left( J^+(\supp \bmu_0) \right). \label{N_PRA}
\end{align}
Clearly, this quantity makes sense only for strictly localised initial states, as if $\supp \mu_0 = \mathbb{R}^n$ and thus $\supp \bmu_0 = \{0\} \times \mathbb{R}^n$, then $N(t,\mu_0) = 0$ for all $t \geq 0$. Also, one should write $N(t,\psi_0)$, with $\mu_0 = \abs{\psi_0}^2$, rather than $N(t,\mu_0)$ to take into account for mean momentum of the initial packet, which does influence its evolution.

In our formalism, the most natural quantification of causality violation is the following
\begin{align}\label{Q1}
\widetilde{N}(t,\psi_0) \vc \inf \{ \omega(\M^2 \setminus J^+) \, \vert \, \omega \in \Pi(\bmu_0, \bmu_t)\} = 1 - \sup \{ \omega(J^+) \, \vert \, \omega \in \Pi(\bmu_0, \bmu_t)\}.
\end{align}
With Definitions \ref{causality_def_true} and \ref{def:causal_evolution} we have $\widetilde{N}(t,\psi_0) = 0$ if and only if $\bmu_0 \preceq \bmu_t$.% \footnote{Note that (\ref{Q1}) is well-defined, because the map $\omega \mapsto \omega(J^+)$ is upper semi-continuous (by the Portmanteau Theorem) and so it attains its maximum on the set $\Pi(\bmu_0, \bmu_t)$, which is narrowly compact in $\PP(\M^2)$ \cite{UsersGuide}.}.

However, equation \eqref{Q1} is not very convenient for concrete computations as one needs to explore the whole space $\Pi(\bmu_0, \bmu_t)$, which is vast. Also, its relationship with the actual possibility of superluminal information transfer is not visible.

Drawing from Theorem \ref{charcond} we can define another measure of causality violation, which mimics, to some extent, the quantity \eqref{N_PRA} defined in \cite{WSWSG11}. Namely, let us set
\begin{align}\label{M}
M(t,\psi_0) \vc \sup \{ M(t,\psi_0,\K) \, \vert \, \K \text{ compact subset of } \supp \, \bmu_0 \},
\end{align}
where
\begin{align}\label{Mk}
M(t,\psi_0,\K) &\vc \max \left\{ 0, \bmu_0(\K) - \bmu_t(J^+(\K)) \right\}.
\end{align}
%Using the fact that $\bmu_0,\bmu_t$ are concentrated on time-slices, $\bmu_t(J^+(K))$ can be rewritten conveniently rewritten in terms of the algebraic sum as $\bmu_t(K+B_t)$, where $B_t$ denotes a closed ball in $\sR^n$ of radius $t$ centred at zero.

%Again, the maximum is always attained, since \ME{Why actually?}. We shall denote by $C_m$ the compact set, which gives the maximum, i.e. if $M(t,\psi_0) > 0$, then $M(t,\psi_0) = M(t,\psi_0,C_m)$.

The number $M(t,\psi_0,\K) \in [0,1]$ can be thought of as %the maximal success rate of the information transfer via a
the `capacity of the superluminal communication channel' -- discussed in Section \ref{subsec:general}. In this context, it is desirable to keep track of the dependence of $M(t,\psi_0,\K)$ on $\K$ to see whether the latter is not unreasonably large (or small) for the information transfer to be possible -- even in principle.

Note, that the difference $\bmu_0(\K) - \bmu_t(J^+(\K))$ cannot, in general, be understood as the `outside probability' \cite{WSWSG11}, i.e. the pure `leak-out' of the probability. The latter holds only if $\K = \supp \, \bmu_0$ is compact. In general, $J^+(\K)$ depends causally on the region $J^-(J^+(\K)) \supseteq \K$, so the flow of probability into $J^+(\K)$ from outside of $\K$ can diminish, or even completely compensate, the visible acausal effect. In fact, the superluminal flow can conspire in such a way that it might be hard in practice to find a compact region $\K \subseteq \{0\} \x \R^n$, for which $M(t,\psi_0,\K) > 0$ for given $t$ and $\psi_0$. Nevertheless, it turns out that in the relativistic-Schr\"odinger system the quantity $M(t,\psi_0)$ helps understanding the acausal behaviour and gives somewhat larger values than $N(t,\mu_0)$ in the limit of a perfectly localised initial state.

%To obtain a `global' number one can simply take, for a given $t$ and $\psi_0$, the supremum of $M(t,\psi,\K)$ over the compact connected achronal $\K \subset \M$.

\subsection{A non-relativistic system}

Let us first consider a non-relativistic quantum system, for which one would expect an acausal behaviour. Indeed, for instance the well-known spreading of the Gaussian wave packet of a free massive quantum particle is acausal in the sense of Definition \ref{def:causal_evolution}. Let us illustrate this fact by considering an initial wave function $\psi(0,x) = (\tfrac{2}{\pi})^{1/4} e^{-x^2}$ evolving on the 2-dimensional Minkowski spacetime with the Hamiltonian $\tfrac{1}{2m} \dt_x^2$. The resulting evolution of probability measures (in natural units) reads
\begin{align*}
d\mu_t(x) = \sqrt{\frac{2}{\pi(1+4 (t/m)^2)}} e^{- \tfrac{2 x^2}{1+4(t/m)^2}}\, dx.
\end{align*}
To show that the evolution $\E: t\mapsto \mu_t$ is acausal we exploit Proposition \ref{charcond}. If we take $K = [-a,a]$ for some $a>0$, then
\begin{align*}
\bmu_t(J^+(\{0\} \times K)) = \mu_t([-a-t,a+t]) = \int_{-a-t}^{a+t} d\mu_t = \mathrm{Erf} \left(\frac{\sqrt{2} m (a+ t)}{\sqrt{m^2+ 4 t^2}}\right),
\end{align*}
where $\mathrm{Erf}$ is the error function. Since the latter increases monotonically, %\TM{Od pewnego momentu?}\ME{Erf zawsze rosnie.},
we conclude that for $a > \frac{m \left(\sqrt{m^2+4 t^2}+m\right)}{4 t}$ we have $\int_{-a-t}^{a+t} d\mu_t < \int_{-a}^{a} d\mu_0$ for every $t>0$. Hence, for any $t>0$ there exists a compact set $\K = \{0\} \times K \subset \R^2$, such that the inequality $\bmu_t(J^+(\K)) < \bmu_0(\K)$ holds and so $\bmu_0 \npreceq \bmu_t$.

We can now proceed to the study of two specific relativistic quantum systems driven by the Dirac and relativistic-Schr\"odinger Hamiltonians. %\ME{Ended here 30.VI}

\subsection{\label{subsec:Dirac}The Dirac system}

Let us first turn to the Dirac system, which is generally believed to conform to the principle of causality \cite{Barat2003,Hegerfeldt2001,WSWSG11}. Below, we confirm this statement in the rigorous sense of Definition \ref{def:causal_sys}.
%Definition \ref{def:causal_evolution}.

\begin{Prop}\label{Ex1}
Let $\psi \in L^2(\R^{1+n}) \otimes \sC^{2^{\lfloor (n+1)/2\rfloor}}$ be a solution to the $(1+n)$-dimensional Dirac equation \footnote{Our conventions are: $\eta = \diag\{-1,1,\ldots,1\}$, $\gamma^\mu \gamma^\nu + \gamma^\nu \gamma^\mu = -2 \eta^{\mu\nu} I_n$, $\left(\gamma^0\right)^\dag = \gamma^0$ and $\left(\gamma^k\right)^\dag = -\gamma^k$ for $k = 1,\ldots,n$.}
\begin{align}
\label{Dirac1}
i \gamma^\mu \partial_\mu \psi - m \psi = 0
\end{align}
\noindent
and let $\psi^\dagger(t,x) \psi(t,x) \, d^n x$ be the corresponding time-dependent probability density. Then, the Dirac system is causal in the sense of Definition
\ref{def:causal_sys}.
\end{Prop}
\begin{proof}
The proof is a straightforward application of Corollary \ref{contcor}. The associated continuity equation is satisfied with $\rho := \psi^\dag \psi$ and $\textbf{j} := \left( \psi^\dag \gamma^0 \gamma^k \psi
\right)_{k=1, \ldots, n}$. In this case, $\rho$ is a probability density function (and so $Q = 1$) and the quantity $J := ( \rho, \textbf{j})$ can be simply written as
\begin{align}\label{Dirac_current}
J^\mu :=  \psi^\dag \gamma^0 \gamma^\mu \psi.
\end{align}
$J$ is well-known to enjoy the transformation properties of a vector field on the $(1+n)$-dimensional Minkowski spacetime.

Moreover, this vector field is causal everywhere. Indeed, assume that $J$ is spacelike at some event $p$. Then, we can find an inertial frame in which $J^{\prime0}(p) = 0$, that is $\psi^{\prime\dag}(p) \psi^\prime(p) = 0$
and therefore $\psi^\prime(p) = 0$. But this would mean that also $\psi(p) = 0$, because $\psi(p)$ and $\psi^\prime(p)$ are related through a unitary transformation. On the other hand, $\psi(p) = 0$ would imply $J(p)=0$ -- a contradiction with the assumption that $J$ was spacelike at $p$.
\end{proof}

%Let us stress that causality of the Dirac evolution in the sense of Definition \ref{def:causal_evolution} has no direct connection with the fact that \eqref{Dirac1} is a hyperbolic equation. The latter indeed implies a finite speed propagation, but it does so at the level of wave functions. The corresponding probability densities include the interference effects, which could in principle spoil the causal flow.

%\ME{Something on hydrodynamical interpretation of Dirac eq}
% The resulting equation for the probability density $\psi^{\dagger} \psi$ is non-linear \ME{Reference needed} and the causality of the evolution is by no means visible from its form.

Let us emphasise the fact that in the Dirac system causality is satisfied during the evolution of \emph{any} initial spinor. In particular, we impose no restrictions on its energy or localisation. This fact does not contradict Hegerfeldt's results (see \cite{Hegerfeldt2001}), as it is well known \cite{Thaller,Barat2003} that positive-energy Dirac wave packets cannot have the localisation properties required by Hegerfeldt's theorem \cite{Hegerfeldt1985}.

We conclude this section with an extension of Proposition \ref{Ex1} to interacting Dirac systems.

\begin{Rem}
The proof of causality of the Dirac system relies on the basic continuity equation
\begin{align}\label{Dirac_cont}
\dt_{\mu} J^{\mu} = 0
\end{align}
enjoyed by the probability current $J^{\mu}$. The latter, as a fundamental law of probability conservation, which holds also in presence of an external electromagnetic or Yang--Mills potentials. In the latter case, the wave function $\psi$ acquires additional degrees of freedom. In general, the Dirac system with \emph{any} interaction which does not spoil the continuity equation \eqref{Dirac_cont} is causal in the sense of Definition \ref{def:causal_sys}.
\end{Rem}

\subsection{\label{subsec:rS}The relativistic-Schr\"odinger system}

We now turn to the relativistic-Schr\"odinger system, i.e. we consider wave packets evolving under the Hamiltonian $\hat{H} = \sqrt{\hat{p}^2 + m^2}$, with $\hat{p} = - i \dt_x$ and $m \geq 0$. For the sake of simplicity, we restrict ourselves to the case of spin 0 representation and one spacial dimension.

Since in the relativistic-Schr\"odinger system $\hat{H} \geq 0$, Hegerfeldt's theorem applies and we expect the evolution of a localised initial state to be acausal. This has been checked (and quantified) in \cite{WSWSG11} for a family of compactly supported initial wave packets $\psi_0(x) = \frac{1}{\sqrt{2d}} \chi_{[-d,d]}(x)$, with $\chi$ being the characteristic function. Because of Proposition \ref{neccond}, this result implies that the evolution of measures in this case is acausal. We consequently conclude that the relativistic-Schr\"odinger system is not causal and thus does not meet Principle~\ref{post}. However, compactly supported states are unphysical idealisations (cf. for instance the Reeh--Schlieder theorem \cite{ReehSchlieder}). Moreover, in the relativistic-Schr\"odinger system the property of compact spacial support is lost whenever the wave packet is boosted to any other frame \cite{WSWSG11}. It is therefore instructive to study the evolution of other classes of initial wave packets to gain better understanding of the nature of causality violation in this system.

Given any initial state $\psi_0 \in L^2(\R)$, the evolution under $\hat{H}$ yields for any $t\geq 0$,
\begin{align}\label{evo}
\psi(t,x) = \frac{1}{\sqrt{2\pi}} \int_{-\infty}^{\infty} \, \widehat{\psi}_0(p) \exp\left( - i \sqrt{p^2 + m^2} t + i p x \right) \, dp\, ,
\end{align}
where $\widehat{\psi}_0$ is the Fourier transform of $\psi_0$.

To check whether the evolution of measures $\E: t \mapsto \mu_t$ with $d\mu_t = \abs{\psi(t,x)}^2 dx$ breaks causality in the sense of Definition \ref{def:causal_evolution} we exploit Proposition \ref{charcond}, similarly as we did for the non-relativistic Hamiltonian. In the relativistic case, explicit formula for the Fourier integral \eqref{evo} is not available, therefore we had to resort to numerical integration. The complete analysis performed with the help of Wolfram Mathematica 10.0.4 is available online \cite{EM2016_Math}, below we summarise its essential points.

The analysis presented below concerns the behaviour of the quantity $M(t,\psi_0,\K_a)$ for $\K_a = \{0\} \x [-a,a]$, with $a>0$ and initial wave packets with zero average momentum. This simplifies the analysis and is sufficient to understand qualitatively the causality violation effects. On the other hand, the quantitative picture is limited by the choice of working with symmetric intervals only. In particular, we obviously have
\begin{align}\label{Mn}
\wt{M}(t,\psi_0) \vc \sup_{a \in \R} M(t,\psi_0,\K_a) \leq M(t,\psi_0).
\end{align}
Note also that the supremum in $M(t,\psi_0)$ can involve disconnected subsets of $\supp \mu_0$. Nevertheless, %the quantity $M(t,\psi_0,\K_a)$ is sufficiently illustrative to understand the nature of the causality breakdown in the relativistic-Schr\"odinger system. Moreover,
the estimate $\wt{M}(t,\psi_0)$, being only a lower bound of $M(t,\psi_0)$, already gives significantly larger values than $N(t,\mu_0)$ of \cite{WSWSG11} in the limit of a perfectly localised initial state.

In \cite{EM2016_Math} we analysed the impact of a non-zero average momentum of the wave packet $\psi_0$ on $M(t,\psi_0,\K_a)$ and found that it does not change the qualitative picture presented below. Note also that a state with a non-zero average momentum can always be boosted to a frame where $\<\hat{p}\> = 0$, what, in view of the discussion following Definition \ref{def:causal_evolution}, will not change the conclusions about the (a)causal behaviour, though it will affect the quantitative picture. In \cite{EM2016_Math} we have also studied the asymmetric case -- with $\K = \{0\} \x [a,b]$. It turns out, not surprisingly, that for symmetric initial wave functions with vanishing average momentum the maximum of $M(t,\psi_0,\{0\} \x [a,b])$ is actually attained for some symmetric interval $[-a,a]$. This is no longer true if the initial wave packet has a nonvanishing expectation value of $\hat{p}$. In the case of $\<\hat{p}\> > 0$, the maximal causality violation is observed by picking the interval $[a,b]$ with $a<0<b$ and $\abs{b} < \abs{a}$. This confirms the supposition that causality breakdown is best visible when the spreading effects are more important than the average motion of the packet.

We shall first focus on the massive case $m>0$ and then briefly comment on the massless one. If $m >0$, we can set $m=1$ without loss of generality. Indeed, note that \eqref{evo} implies
\begin{align*}
\psi(t,x,\psi_0;m) = \psi(mt,mx,\psi_0(\cdot /m);1),
\end{align*}
hence
\begin{align}\label{Mm}
M(t,\psi_0,\{0\} \x [a,b];m) & = M(mt,\psi_0(\cdot /m),\{0\} \x [a/m,b/m];1), \notag \\
M(t,\psi_0;m) & = M(mt,\psi_0(\cdot /m);1).
\end{align}

% Gaussians

The first class of initial states in the relativistic-Schr\"odinger system that we have analysed in detail are the Gaussian wave packets
\begin{align}\label{gaussian}
\psi^G_0(x;d) = (\pi d)^{-1/4} \exp\left( \tfrac{-x^2}{2d} \right),
\end{align}
with the width $\sqrt{d} > 0$. % and initial momentum $p_0 \in \R$.

Figure \ref{fig:Gauss} illustrates the behaviour of the quantity $M(t,\psi_0,\K_a)$, with $\K_a = \{0\} \x [-a,a]$ and $d = 1$.

\begin{figure}[h]
\begin{center}
\resizebox{!}{290pt}{\includegraphics[scale=0.6]{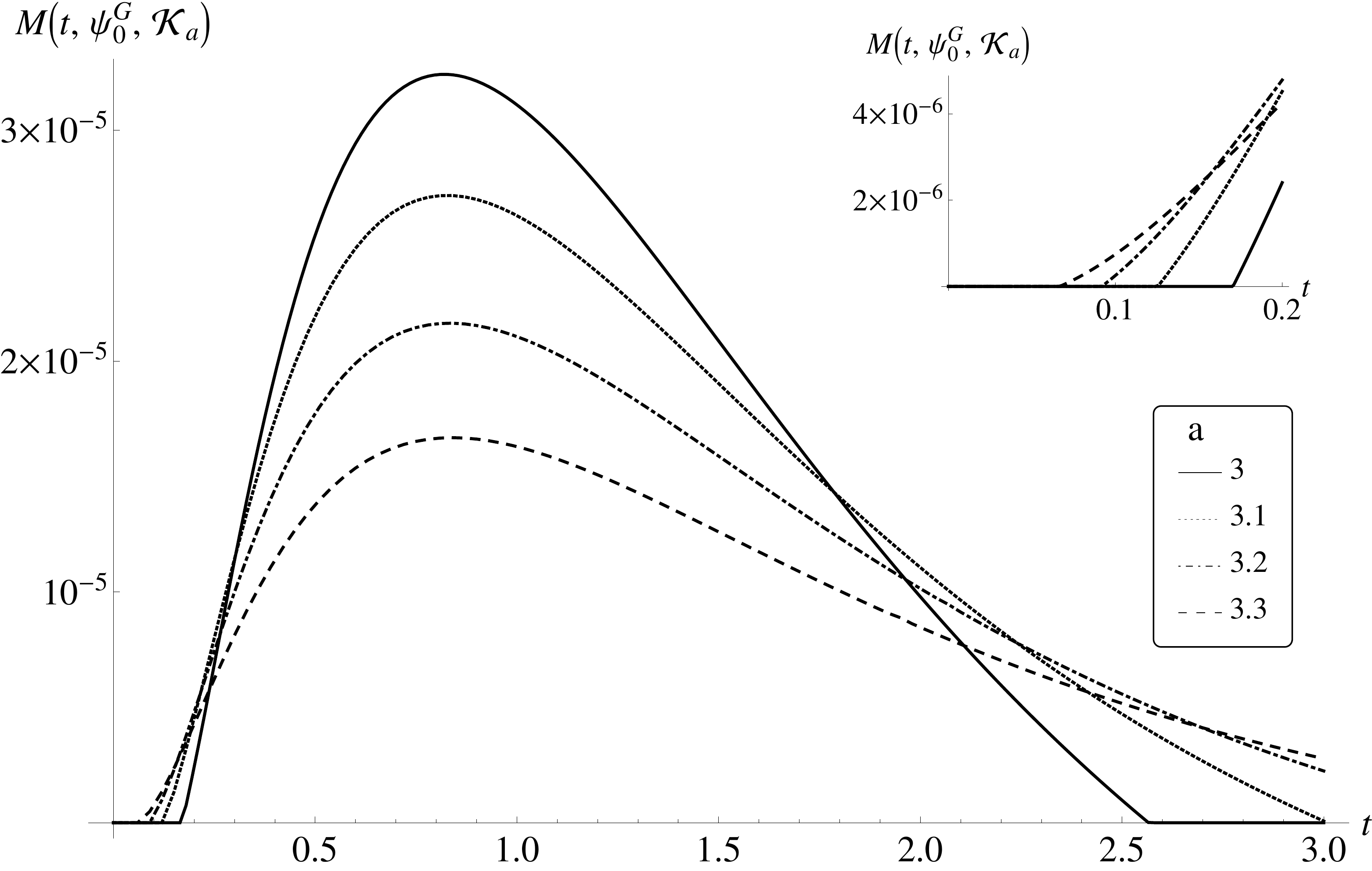}}
\caption{\label{fig:Gauss}Acausal evolution of a Gaussian probability density in the relativistic-Schr\"odinger system.}
\end{center}
\end{figure}

\begin{figure}[h]
\begin{center}
\resizebox{!}{250pt}{\includegraphics[scale=0.6]{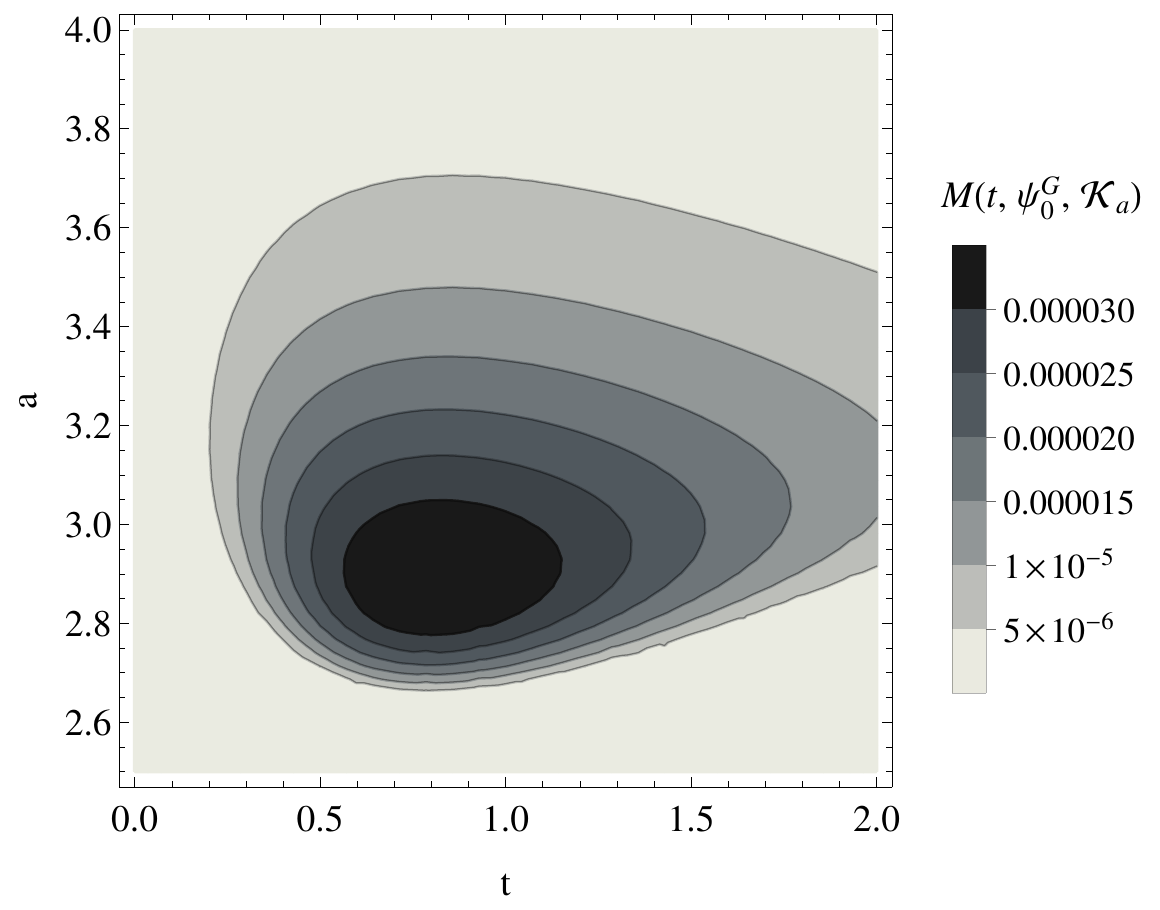}}
\caption{\label{fig:Gauss2} Estimation of the values of the parameters $t_1$ and $a_M$ in the massive ($m=1$) case.}
\end{center}
\end{figure}

At first, the quantity $M(t,\psi_0,\K_a)$ is zero suggesting a causal evolution. Then, for some $t=t_0$, it starts increasing, manifesting the breakdown of causality. For later times ($t > t_1$), the probability flow `slows down' and the quantity $M(t,\psi_0,\K_a)$ can even decrease to 0 for $t > t_2$ and a suitably chosen compact set $\K$.

In \cite{EM2016_Math} we studied the dependence of the values of time instants $t_0, t_1$ and $t_2$ on the choice of the `size' of the compact set $\K_a = \{0\} \x [-a,a]$, as parametrised by $a$.  It leads to the following conclusions:
\begin{itemize}
\item For $a$ small enough, the quantity $M(t,\psi_0,\K_a)$ is zero for all times and the breakdown of causality is not visible. On the contrary, the values of $a$ larger than $a_0 \approx 2.65$ lead to the acausal behaviour as illustrated in Figure \ref{fig:Gauss}.

\item The first time-scale $t_0$ decreases with larger values of $a$. It suggests that, as in the non-relativistic case, for any $t>0$ there exists a compact set $\K \subset \{0\} \times \R$, such that the inequality $\bmu_t(J^+(\K)) < \bmu_0(\K)$ holds and thus causality is actually broken immediately once the evolution starts.

\item On the other hand, the scale of causality breakdown, quantified by \eqref{Mn}, becomes smaller for larger regions $\K_a$. It attains a maximum $\wt{M}(t,\psi_0^G) = 3.55 \times 10^{-5}$ for $t_1 = 0.81$ and $a_M = 2.89$ -- see Figure \ref{fig:Gauss2}.

\item The causality breakdown has a transient character quantified by the time-scale $t_1 (\K_a) = \argmax_{t \geq 0} M(t,\psi_0,\K_a)$. The quantity $t_1 (\K_a) \approx 0.8$ does not depend significantly on the choice of $a$, provided $a > a_0$.

\item The third time-scale $t_2$, capturing the restoration of causality, can be made arbitrarily large by choosing $a$ large enough.

%\item Allowing for $p_0 \neq 0$ does not change qualitatively the pictures \ref{fig:Gauss} and \ref{fig:Gauss2}.

\end{itemize}

With the narrowing of the initial Gaussian width $d$, the quantity $\wt{M}(t,\psi_0^G;d)$ grows, whereas the time-scale $t_1$ decreases slightly,  as illustrated by the following table:
\begin{gather*}
\begin{array}{|c|c|c|c|c|c|c|}
\hline
d & 1 & 10^{-1} & 10^{-2} & 10^{-3} & 10^{-4} & 10^{-5} \\
\hline
\wt{M}(t_1) & 0.000035 & 0.0066 & 0.039 & 0.079 & 0.106 & 0.121 \\
\hline
t_1 & 0.81 & 0.68 & 0.64 & 0.58 & 0.53 & 0.48 \\
\hline
a_M & 2.89 & 0.63 & 0.165 & 0.048 & 0.015 & 0.0047 \\
\hline
\end{array}
\end{gather*}

In the limit $d \to 0$, the quantity $\wt{M}(t,\psi_0^G)$ tends to the maximum of approx. 0.13. This value is by 60\% larger than the maximal `outside probability' computed in \cite{WSWSG11}. %This discrepancy stems from the fact that in the case of compactly supported initial states, the set $\K_M$, which gives the maximum in $\wt{M}(t,\psi_0^G)$ always equals to $\supp \psi_0$, whereas for the Gaussian states the width of $\K_M = [-a_M,a_M]$ is not related in a simple way to the Gaussian width $\sqrt{d}$.
It shows, that to quantify the amount of the causality breakdown for arbitrary wave packets it is not sufficient to look at one specific region of space from which the probability `leaks too fast'.

In the massless case, the causality breakdown in the quantum system driven by the Hamiltonian $\hat{H} = \sqrt{\hat{p}^2}$ has a persistent rather than transient character: The quantity $\wt{M}(t,\psi_0^G)$ is greater than 0 for any $t>0$ and increases monotonically -- see Figure \ref{fig:Gauss0}. It approaches asymptotically the value 0.13, in consistency with the above results and formula \eqref{Mm}.

\begin{figure}[h]
\begin{center}
\resizebox{!}{250pt}{\includegraphics[scale=0.6]{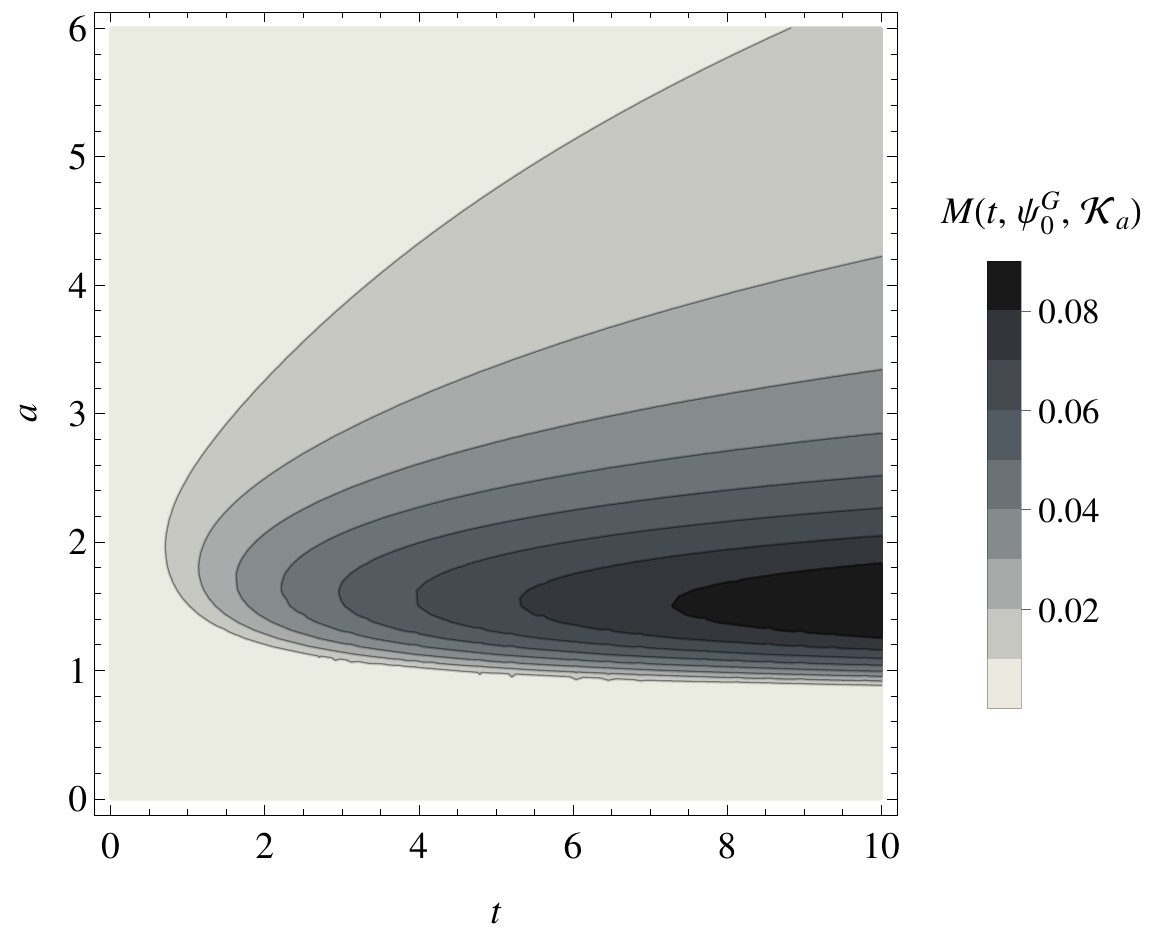}}
\caption{\label{fig:Gauss0} The persistent character of causality violation in the massless case.}
\end{center}
\end{figure}

% Exponential tails 1

Let us now return to the massive case $m=1$ and analyse the second class of initial states with exponentially bounded tails,
\begin{align}\label{exp}
\psi_0^e(x) = \sqrt{\tfrac{\alpha}{2}} \sech( \alpha x),
\end{align}
for $\alpha > 0$. % and the initial momentum $p_0 \in \R$.
Thanks to the fact that $\sech$ is its own Fourier transform, the states \eqref{exp} have exponential tails also in the momentum representation, which makes them suitable for numerical integration.

According to Hegerfeldt's result, one expects an acausal evolution for $\alpha > m =1$. The following table illustrates the amount of causality violation quantified by formula \eqref{Mn} as $\alpha$ approaches the Hegerfeldt's bound.

\begin{gather*}
\begin{array}{|c|c|c|c|c|c|}
\hline
\alpha & 3 & 2 & 5/3 & 3/2 \\
\hline
\wt{M}(t_1) & 3 \times 10^{-4} & 2 \times 10^{-6} & 1.4 \times 10^{-8} & 10^{-10}\\
\hline
t_1 & 0.79 & 0.83 & 0.84 & 0.85\\
\hline
a_M & 1.4 & 3.2 & 5.2 & 7.4 \\
\hline
\end{array}
\end{gather*}

As $\alpha$ tends to $\infty$ one obtains a maximal amount of causality violation around 13\%. This is consistent with the result we obtained above for the $\delta$-like limit of the initial Gaussian states.

On the other hand, the amount of causality violation decreases fast as $\alpha$ approaches $m=1$. It suggests that the evolution of measures triggered by the initial state \eqref{exp} with $\alpha = m = 1$ is causal. Indeed, in \cite{EM2016_Math} we found no evidence of causality violation during the evolution of such an initial wave packet.

This observation is, however, only an artefact of the chosen class of states. The next example shows that the Hegerfeldt's bound is in fact artificial.

% Exponential tails + sin

We now investigate the evolution of initial states
\begin{align}\label{sinexp}
\psi_0^{SE}(x) = \mathcal{N} \, \frac{\sin x}{x} \sech( \alpha x),
\end{align}
for $\alpha >0$, with the normalisation constant $\mathcal{N}$. States in this class still have exponentially bounded tails both in position and momentum representation.

By computing the quantity $\wt{M}(t,\psi_0^{SE})$ we found in \cite{EM2016_Math} a clear evidence of causality violation for all values of $\alpha \in [0,4]$, as shown on Figure \ref{fig:exp}.

\begin{figure}[h]
\begin{center}
\resizebox{!}{280pt}{\includegraphics[scale=0.6]{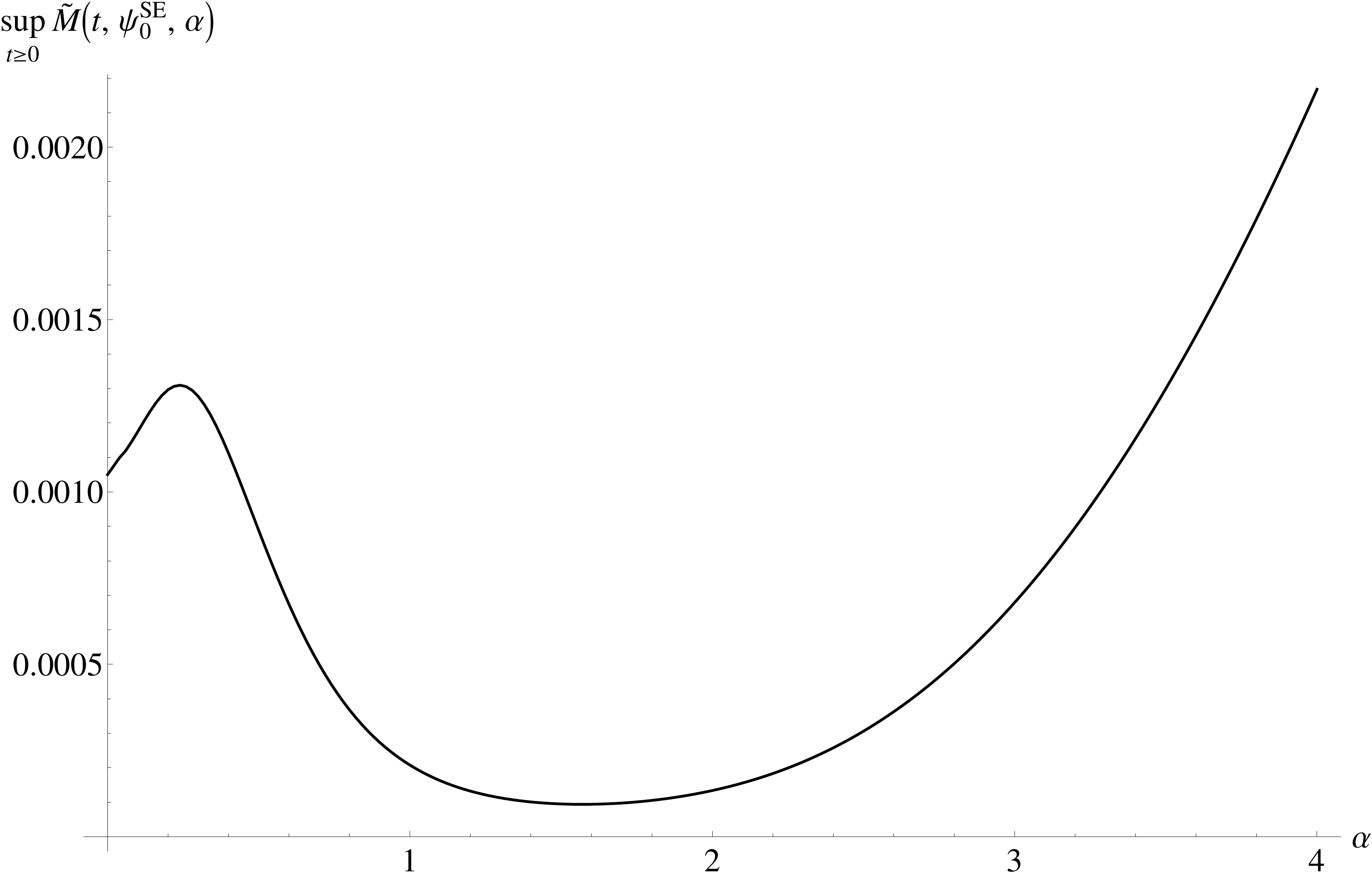}}
\caption{\label{fig:exp}The maximal amount of causality violation during the evolution of initial states in the class \eqref{sinexp}.}
\end{center}
\end{figure}

We see that initial states decaying as $e^{-m \norm{x}}$ play no special role in the causality violation effects in the relativistic-Schr\"odinger system. Although there seems to be local minimum for $\alpha \approx 1.5$, which may well be an artefact of the fact that $\wt{M}(t,\psi_0^{SE})$ is only a lower bound of $M(t,\psi_0^{SE})$, the quantity $\wt{M}$ is manifestly positive for all $\alpha$.

% Power-like tails

In particular, note that causality of the evolution is spoiled for initial states in the class \eqref{sinexp} with $\alpha = 0$, which decay only as $\mathcal{O}(x^{-1})$. In fact, our numerical analysis suggests that the breakdown of causality is generic also in a wider class of states with heavy tails:
\begin{align*}
\psi_0^{S}(x) = \mathcal{N} \, \left(\frac{\sin p_m x}{p_m x}\right)^n,
\end{align*}
with $n \in \mathbb{N}$ and $p_m >0$. In \cite{EM2016_Math} we checked it explicitly for $n \in \{1,2,3\}$ and $\tfrac{1}{10} \leq p_m \leq 10$.

%\ME{Case $m=0$}.

Let us now now summarise our analysis, draw conclusions and compare them with the controversial upshot of Hegerfeldt.

\section{\label{sec:outlook}Discussion}

\subsection{\label{subsec:Hegerfeldt}The claim of Hegerfeldt}

To facilitate the comparison let us first briefly summarise Hegerfeldt's results on causality presented in \cite{Hegerfeldt1985} and his other works  \cite{Hegerfeldt1,HegerfeldtFermi,Hegerfeldt1998_ann,Hegerfeldt2001,Hegerfeldt2}. We find it important to clarify the field, as the outcomes of \cite{Hegerfeldt1985} are sometimes misinterpreted or overinterpreted (see below). %, even by Hegerfeldt himself (see below).

Hegerfeldt's conclusion concerning the acausal behaviour of the wave packets relies on three assumptions \cite{Hegerfeldt1985}:

\begin{enumerate}

\item For any region of space $V \subseteq \R^3$, there exists a positive operator $N(V) \leq 1$, such that $\< \psi | N(V) | \psi \>$ yields the probability of finding in $V$ a particle in the state $\psi$.

\item The evolution of the system is driven by a positive Hamiltonian operator $\hat{H}$.

\item There exists a state $\psi_0$ with exponentially bounded tails, i.e.
\begin{align}\label{exp_bound}
\< \psi_0 | N(\mathbb{R}^3 \setminus B_r) | \psi_0 \> \leq K_1 \exp(-K_2 r^k), && \text{ for sufficiently large $r$},
\end{align}
where $B_r$ is a closed ball of radius $r$ centred at the origin.

\end{enumerate}

The constants $K_1, K_2$ depend on $\psi_0$ and the exponent $k$ depends on the chosen Hamiltonian. More concretely, one has \cite{Hegerfeldt1985} $k=1$, $K_2 > m$ for the free relativistic-Schr\"odinger Hamiltonian $\hat{H} = \sqrt{\hat{p}^2 + m^2}$ and $k=2$, $K_2$ arbitrarily small for more general systems with interactions.

Under the above assumptions, Hegerfeldt obtained the following result:

\begin{Thm}[Hegerfeldt Theorem \cite{Hegerfeldt1985}]\label{thm:Heger}
In the quantum system fulfilling the assumptions (1) and (2) let $\psi_0$ be a state satisfying \eqref{exp_bound}. Then,
\begin{align}\label{Heger}
\forall \, t > 0 \ \ \exists \, \ba \in \mathbb{R}^3 \ \ \exists \, r > 0 \ \quad \left\langle \psi_t | N(B_{\ba,r}) | \psi_t \right\rangle > \left\langle \psi_0 | N(\mathbb{R}^3 \setminus B_{\| \ba \| -r-t}) | \psi_0 \right\rangle,
%\left\langle \psi_0 | N(B_{\ba,r+t}) | \psi_0 \right\rangle,
\end{align}
where $B_{\ba,r}$ denotes a closed ball of radius $r$ centred at $\ba$.
\end{Thm}
Let us stress that, although condition \eqref{Heger} is never mentioned explicitly in Hegerfeldt's works, it is this condition which is actually proven in \cite{Hegerfeldt1985}.

In the original formulation, Hegerfeldt demonstrated the above result under the assumption of arbitrary finite propagation speed $c'$. However, since the strict inequality \eqref{Heger} holds for any $t>0$, we can set $c' = 1$ without loss of generality.

Since it is obviously true that $\mathbb{R}^3 \setminus B_{\|\ba\|-r-t} \supseteq B_{\ba,r+t}$, therefore \eqref{Heger} implies that
\begin{align}\label{Heger2}
\forall \, t > 0 \ \ \exists \, \ba \in \mathbb{R}^3 \ \ \exists \, r > 0 \ \quad \left\langle \psi_t | N(B_{\ba,r}) | \psi_t \right\rangle > \left\langle \psi_0 | N(B_{\ba,r+t}) | \psi_0 \right\rangle.
\end{align}
This result, albeit somewhat weaker than \eqref{Heger}, has a clearer interpretation. Namely, it shows that for any $t > 0$ \emph{there exists a ball} in $\mathbb{R}^3$, into which the probability `has been leaking too fast' by the time $t$ has elapsed.

We emphasise the ``there exists a ball'' phrase in the above results. This makes them considerably weaker statements than the one alleged by Hegerfeldt in \cite{Hegerfeldt1998_ann}, where the author announces the superluminal flow of probability from \emph{any} ball centred at the origin. The latter claim is in fact false in the relativistic-Schr\"odinger system, as we have seen in the previous Section. Additionally, notice that \eqref{Heger2} speaks about the inflow of probability into a ball rather than the outflow.

\subsection{Summary of the obtained results}

In our study of causality in quantum mechanics we have followed a different path than Hegerfeldt, although the underlying concept is quite similar. Our Definition \ref{def:causal_evolution} agrees with the viewpoint on causality in quantum mechanics, %\cite{Berry2012}, %\ME{Cite to be expanded}
 shared in particular by Hegerfeldt, in that it should be about the \emph{flow of probability}. % rather than the `localisation' properties of quantum states.
We claim that the property of being causal or not should refer to the physical system (or, more precisely, to the theory modelling the system at hand) and not to some particular class of its states. One of the motivations behind such a view is the fact that whereas the spatial properties of wave packets in relativistic quantum systems depend on the chosen frame \cite{Haag}, the causality of evolution of measures is an observer-invariant concept (cf. \cite[Section 4]{Miller16}).

Our study of the relativistic-Schr\"odinger system supports the above claim. We have shown that the superluminal flow of the probability density is not related to the decay-in-space properties of the initial wave packet. In particular, Hegerfeldt's assumption \eqref{Heger} seems to be merely an artefact of his technique of proving Theorem \ref{thm:Heger}. This feature constitutes the first difference between Hegerfeldt's approach and ours: we do not make any assumptions about the form of the wave packets.

The second advantage of our formalism consists in the fact that we do not need to assume the positivity of energy. The latter assumption does play an important role in the wave packet formalism, as, for instance, positive-energy solutions of the Dirac equation cannot satisfy Hegerfeldt's bound \eqref{Heger} \cite{Barat2003,Thaller}. However, it does not seem to influence the (a)causality of the probability flow. Our result (Proposition \ref{Ex1}) %attesting the causal evolution of measures in the Dirac system
shows that bizarre phenomena resulting from the interference of positive and negative frequency parts of the packet \cite{Thaller2dim}, such as Zitterbewegung \cite{ZitterSchroedinger}, do not spoil the causal evolution of probabilities.

The third characteristic of our strategy is that we do not require \emph{a priori} the existence of any position operator, although in Section \ref{sec:QM} we implicitly assumed that the probability measures are calculated from wave functions via the usual (often named `non-relativistic' \cite{Thaller}) position operator $(\hat{x} \psi)(x) = x \psi(x)$. We did so firstly to facilitate the comparison of our results with the conclusions of \cite{WSWSG11} and, secondly, because $\abs{\psi}^2$ is in fact an observable quantity, which can be measured experimentally (see for instance \cite{ZitterNature}). If one chooses to work, for instance, with the Newton--Wigner position operator $\hat{x}_{\mathrm{NW}}$, one can re-express the probability measure obtained with $\hat{x}_{\mathrm{NW}}$ in terms of the standard `modulus square principle' via the Foldy--Wouthuysen transformation \cite{Thaller}. The corresponding transformed wave packets can never have compact spacial supports, but the flow of probability remains causal on the strength of Proposition \ref{Ex1}.

\subsection{Outlook}

The philosophy behind Definition \ref{def:causal_evolution} is to consider probability measures on spacetime, which model the outcomes of some experiment.

In the general framework outlined in Section \ref{sec:omega} we do not have to ask where do the measures actually come from. Principle \ref{post} states, however, that regardless of the procedure that leads to an evolution of measures at hand, the latter needs to be causal in the rigorous sense of Definition \ref{def:causal_sys}. In the context of the wave packet formalism this postulate accords with the viewpoint that wave functions are not physical objects -- they are just a way to compute probabilities \cite{QIandGR}. We claim that if Principle \ref{post} is violated for some system, it means that the model which yields the dynamics of probability measures is inadequate. More precisely, if causality violation effects, as quantified with the help of the tools from Section \ref{subsec:quantify}, are significant within the domain of applicability of the model, then the model has to be discarded. Let us note that a similar principle was applied in \cite{QFT_detection} to demonstrate the advantage of the Unruh--deWitt model of detection in quantum field theory over the popular Glauber scheme.

From the empirical point of view, Principle \ref{post} implies that the superluminal flow of probability cannot be observed in any experiment. In this spirit one could use it to discriminate various hidden variables theories, also the non-local ones \cite{Zeilinger_nonlocal}, as well as theories with correlations stronger than quantum \cite{InformationCausality}. On the other hand, one can look for evidence of (the analogues of) causality violation effects in a suitable quantum simulation \cite{QuantumSimulation} of the relativistic-Sch\"odinger Hamiltonian.  Let us also note that, within our general formalism, one can incorporate into the measures the errors resulting from the measuring apparatus' imperfections, including the time measurement, or dark counts caused by the quantum vacuum excitations.

As stressed in the Introduction, we consider the wave packet formalism as a phenomenological description of quantum systems, which actually require a quantum field theoretic model. In fact, we regard the probability measures on a spacetime $\M$ as mixed states on the \emph{commutative} $C^*$-algebra of observables $C_0(\M)$. They can thus be seen as outcomes of a channel transforming quantum information into the classical one -- an observable, or more generally an instrument \cite{Keyl}. The measures can thus result from multi-particle quantum states, modelling the effective density of the atomic cloud subject to a direct detection, for instance in the Bose--Einstein condensate \cite{Sacha}.

Let us conclude with an outline of the potential extensions and future application of the developed formalism.

Since the framework of \cite{EcksteinMiller2015} is generally covariant, it seems natural to envisage an extension of the outcomes of Section \ref{sec:QM} to curved spacetimes. The wave packet formalism in the external gravitational field (see for instance \cite{MinimalPacket}) is particularly useful in the study of neutrino oscillations \cite{Beuthe,BernardiniFlavour}. Such an extension, which would require a covariant continuity equation for measures is, however, not that straightforward. The stumbling block is Theorem \ref{superposition} in the optimal transport theory, which is formulated only on $\R^n$.

Another desirable application would be to consider signed measures. This would open the door to the study of causality in the Klein--Gordon system, where the density current does not have a definite sign. A more radical extension would consist in extending the causal relation onto the space $\PP(\M,\B(\H))$ of Borel probability measures on spacetime $\M$ with values in a, possibly noncommutative, algebra of observables $\B(\H)$. Definition \ref{causality_def_true} can be easily adapted to this case: condition i) stays unaltered, whereas the second requirement will take the form $\omega(J^+) = \id_\H$. The details of such a construction, in particular an analogue of Theorem \ref{charcond}, require more care an are to be unravelled. In this framework, one could construct POVM's on $\M$ with the spacetime events regarded as possible outcomes of a generalised observable. With a definite causal order on $\PP(\M,\B(\H))$ one might be able to address the pertinent problem \cite{QuantumCausality} of finding a unified framework for the study of quantum correlations between spacelike and timelike separated regions of spacetimes.

\begin{acknowledgments}
We are grateful to Pawe{\l} Horodecki and Marcin P{\l}odzie\'n for numerous enlightening discussions. We also thank Henryk Arod\'z for comments on the manuscript. ME was supported by the Foundation for Polish Science under the programme START 2016. ME acknowledges the support of the Marian Smoluchowski Krak\'ow Research Consortium ``Matter--Energy--Future'' within the programme KNOW.
\end{acknowledgments}

\bibliography{causality_bib}

\end{document}